%% file: ms.tex
\documentclass[letterpaper,11pt]{article}

\def\ShowComment{True} %

\input{header}

\title{Cactus Representation of Minimum Cuts:\\ Derandomize and Speed up}

\author{
Zhongtian He\\
Princeton University
\and
Shang-En Huang\thanks{Supported by NSF Grant No. CCF-2008422.}\\
Boston College
\and
Thatchaphol Saranurak\thanks{Supported by NSF CAREER grant 2238138.}\\
University of Michigan
}

\date{}

\begin{document}

\begin{titlepage}
    \thispagestyle{empty}
    \maketitle
    \begin{abstract}
        \thispagestyle{empty}
        \input{0-abstract}

    \end{abstract}
\end{titlepage}

\renewcommand{\baselinestretch}{0.7}\normalsize
\tableofcontents
\renewcommand{\baselinestretch}{1.0}\normalsize
\thispagestyle{empty}
\newpage
\addtocontents{toc}{\protect\thispagestyle{empty}} 
\setcounter{page}{1}

\input{1-introduction}

\input{2-preliminary}

\input{3-comparable}

\input{3-main-new}
\input{4-incomparable}

\input{put-together}

\bibliographystyle{alpha}
\bibliography{ref}

\appendix

\input{A-app-construction}

\input{B-vertexlabel}

\input{C-omit}

\end{document}

%% file: header.tex
\usepackage{fullpage}
\usepackage[utf8]{inputenc}
\usepackage{amsmath,amsfonts,amsthm,amssymb}
\usepackage{setspace}
\usepackage{fancyhdr}
\usepackage{lastpage}
\usepackage{extramarks}
\usepackage{chngpage}
\usepackage{soul,color}
\usepackage{graphicx,float,wrapfig,lipsum}
\usepackage{listings}
\usepackage{xcolor}      %
\usepackage{enumerate}
\usepackage{enumitem}
\usepackage{xspace}
\usepackage[T1]{fontenc}

\usepackage{xcolor}
\usepackage{nameref}
\definecolor{ForestGreen}{rgb}{0.1333,0.5451,0.1333}
\definecolor{DarkRed}{rgb}{0.65,0,0}
\definecolor{Red}{rgb}{1,0,0}
\usepackage[linktocpage=true,
pagebackref=true,colorlinks,
linkcolor=DarkRed,citecolor=ForestGreen,
bookmarks,bookmarksopen,bookmarksnumbered]
{hyperref}
\usepackage{cleveref}

\usepackage{indentfirst}
\usepackage{algpseudocode}
\usepackage{algorithm}  
\usepackage{algorithmicx} 
\usepackage{verbatim}
\usepackage{todonotes}
\usepackage{appendix}

\input{macro}

\makeatletter
\def\namedlabel#1#2{\begingroup
    #2%
    \def\@currentlabel{#2}%
    \phantomsection\label{#1}\endgroup
}
\makeatother

\usepackage{thmtools}

\declaretheorem[numberwithin=section]{theorem}
\declaretheorem[numberlike=theorem]{lemma}
\declaretheorem[numberlike=theorem,name=Lemma]{lem}
\declaretheorem[numberlike=theorem]{fact}

\declaretheorem[numberlike=theorem]{corollary}

\declaretheorem[numberlike=theorem,style=definition]{definition}
\declaretheorem[numberlike=theorem,name=Definition,style=definition]{defn}

\usepackage{mathrsfs}  
\newcommand{\collection}{\mathscr{C}}
\Crefname{ALC@unique}{Line}{Lines}

\ifdefined\ShowComment

\def\thatchaphol#1{\marginpar{$\leftarrow$\fbox{T}}\footnote{$\Rightarrow$~{\sf\textcolor{purple}{#1 --Thatchaphol}}}}

\def\zhongtian#1{\marginpar{$\leftarrow$\fbox{Z}}\footnote{$\Rightarrow$~{\sf\textcolor{blue}{#1 --Zhongtian}}}}

\else
\def\thatchaphol#1{}
\def\zhongtian#1{}

\fi

%% file: macro.tex
\newcommand{\NULL}{\ensuremath{\mathsf{null}}}
\newcommand{\hatP}{\hat{P}}
\newcommand{\T}{\mathcal{T}}
\newcommand{\C}{\mathcal{C}}
\newcommand{\CC}{\mathcal{C}^{\parallel}}
\newcommand{\CI}{\mathcal{C}^{\perp}}

\newcommand{\LCA}{\mathrm{LCA}}

\newcommand{\Maximal}{H}

\newcommand{\cP}{\mathcal{P}}
\newcommand{\cL}{\mathcal{L}}

\newcommand{\AddP}{\texttt{AddPath}}
\newcommand{\MinP}{\texttt{MinPath}}
\newcommand{\MinPDown}{\ensuremath{\texttt{MinPath}^\downarrow}}
\newcommand{\MinPUp}{\ensuremath{\texttt{MinPath}^\uparrow}}
\newcommand{\MinTree}{\texttt{MinTree}}
\newcommand{\MinTreeDown}{\ensuremath{\texttt{MinTree}^\downarrow}}
\newcommand{\MinTreeUp}{\ensuremath{\texttt{MinTree}^\uparrow}}
\newcommand{\Link}{\texttt{Link}}
\newcommand{\Cut}{\texttt{Cut}}
\newcommand{\MaxNonPath}{\texttt{MaxNonPath}\xspace}
\newcommand{\MinNonPath}{\texttt{MinNonPath}\xspace}
\newcommand{\LocalUpdate}{\texttt{LocalUpdate}\xspace}

\newcommand{\cut}{\texttt{Cut}}
\newcommand{\link}{\texttt{Link}}

\newcommand{\upath}[2]{\overline{{#1}{#2}}}
\newcommand{\PreOrder}{\texttt{pre\_order}}

\newcommand{\InsertBST}{\texttt{InsertBST}}
\newcommand{\ExtractBST}{\texttt{ExtractBST}}

\global\long\def\highest{\mathrm{top}}
\global\long\def\lca{\textsc{lca}}

%% file: 0-abstract.tex
Given an undirected weighted graph with $n$ vertices and $m$ edges, we give the first
deterministic $m^{1+o(1)}$-time algorithm for constructing the cactus representation of \emph{all} global minimum cuts.
This improves the current 
$n^{2+o(1)}$-time state-of-the-art deterministic algorithm,
which can be obtained by combining ideas implicitly from three papers
\cite{karger2000minimum, li2021deterministic, gabow2016minset}.
The known explicitly stated deterministic algorithm has a runtime of
$\tilde{O}(mn)$ 
\cite{fleischer1999building,nagamochi2000fast}.
Using our technique, we can even speed up the fastest randomized algorithm of
\cite{karger2009near} whose running time is at least $\Omega(m\log^4 n)$ to $O(m\log^3 n)$.

%% file: 1-introduction.tex
\section{Introduction}

The global minimum cut problem has been studied for decades. For an undirected, weighted graph $G = (V,E,w)$, the global mincut of $G$ is a minimum weight subset of edges that disconnecting the graph by removing them. Lots of beautiful works on this problem appeared in the last century \cite{gomory1961multi, hao1992faster, nagamochi1992computing, nagamochi1992linear, stoer1997simple}, and then a series of work utilizing randomization \cite{karger1993global, karger1996new} finally led to an near-linear time \emph{Monte Carlo} algorithm by Karger \cite{karger2000minimum} in 1996. %

It turns out that there is a cactus representation of \emph{all} (possibly $\Theta(n^2)$) minimum cuts using an $O(n)$-edge cactus graph introduced by \cite{dinits1976structure} (see also \cite{Fleiner2009AQP}). The cactus representation of global mincuts has found several algorithmic applications: it is a key subroutine for several edge connectivity augmentation algorithms \cite{gabow1991applications,naor1997fast,cen2022augmenting, ravi2023approximation} and also in several dynamic mincut algorithms \cite{henzinger1995approximating,goranci2018incremental}. 
Many algorithms were designed to find the cactus representation of minimum cuts. \cite{karzanov1986efficient} outlined the first algorithm for constructing the cactus that takes $\Theta(n^3)$ time.
Their algorithm was parallelized by \cite{naor1991representing} and refined by \cite{Nagamochi1994CanonicalCR}. Later on, faster algorithms were developed by \cite{gabow1991applications,karger1996new,nagamochi2000fast} and \cite{fleischer1999building} where the latter two algorithms running in $\Tilde{O}(mn)$ are the fastest explicitly stated deterministic algorithms for computing cactus representation.
Finally, \cite{karger2009near} showed that the cactus representation problem is near-linear time computable by randomized algorithms as well.

All these near-linear time algorithms mentioned so far have one drawback: they are \emph{Monte Carlo} meaning that they can err.
It was a big open problem if there are near-linear (or almost-linear\footnote{As a convention from the literature~\cite{karger2000minimum,li2021deterministic}, we say a function that is $\tilde{O}(m)$ to be \emph{near-linear} and $m^{1+o(1)}$ to be \emph{almost-linear}.}) time deterministic algorithms for computing global mincuts and even computing cactus.
After a series of works \cite{kawarabayashi2018deterministic,henzinger2020local,saranurak2021simple} and \cite{li2020deterministic}, Li recently showed %
 an $m^{1+o(1)}$-time deterministic
algorithm for computing a global minimum cut \cite{li2021deterministic}, by derandomizing the construction of the skeleton graph which is the single randomized procedure in Karger's near-linear time mincut algorithm \cite{karger2000minimum}.

A natural question is that, given that we can deterministically compute a global mincut, whether we can also compute a cactus representation for all global mincuts as well.
It turns out that Li's deterministic skeleton graph
 construction~\cite{li2021deterministic} applied to Karger and Panigrahi's algorithm~\cite{karger2009near} is not sufficient.
Instead, one may replace the procedure of \cite{karger2009near} by Gabow's algorithm~\cite{gabow2016minset} together with Karger's dynamic programming technique~\cite{karger2000minimum}, achieving a deterministic $n^{2+o(1)}$-time.
Recently, Kawarabayashi and Thorup~\cite{kawarabayashi2018deterministic} and Lo, Schmidt, and Thorup~\cite{lo2020compact} successfully showed how to construct cactus deterministically in near-linear time, \emph{assuming} simple graphs. Nonetheless, the general problem of whether such an algorithm exists for general weighted graphs is still open. 

In this work, we improve the quadratic $n^{2+o(1)}$-time 
barrier by
showing the first almost-linear deterministic algorithm for computing the cactus representation of minimum cuts for undirected weighted graph, 
which positively answers the open question raised by \cite{karger2009near}\footnote{Actually, they asked if there exists an efficient algorithm to (deterministically) compute a \emph{certificate} of a cactus representation, which turns their Monte Carlo algorithm into a Las Vegas  algorithm.}. Furthermore, using our technique, we also speed up the best previous randomized algorithm \cite{karger2009near} by a logrithmic factor.

\begin{theorem}
    \label{thm:main}
    There are algorithms for computing cactus representation of all (global) minimum cuts in an undirected weighted graph with the following guarantees
    \begin{itemize}[itemsep=0pt]
        \item Deterministic algorithm in $m^{1+o(1)}$ time. 
        \item Randomized Monte-Carlo algorithm in $O(m\log^{3}n)$ time.     
    \end{itemize}
\end{theorem}

This is the first almost-linear time deterministic algorithm for computing
cactus. Previously, the fastest deterministic algorithm takes $\Omega(mn)$ by \cite{gabow2016minset,fleischer1999building,nagamochi2000fast}. 
The $m^{o(1)}$ factor in our running time is solely from the overhead of $n^{o(1)}$ in the deterministic tree packing algorithm by Li \cite{li2021deterministic}. If this factor was  improved to $O(\text{polylog}(n))$, our running time would be $\tilde{O}(m)$ too.

By plugging a faster randomized tree packing \cite{karger2000minimum} into our new approach, 
we obtain a randomized algorithm that is even faster than the fastest known algorithm by \cite{karger2009near} which takes time at least $\Omega(m\log^{4}n)$. We discuss this in more details in \Cref{sec:comparison-with-KP09}.

\paragraph{New Development.} After we have published the manuscript, a new result by Henzinger, Li, Rao, and Wang \cite{HLRW2024} shows how to deterministically construct in $O(m\,\text{polylog}(n))$ time a collection of $\text{polylog}(n)$ trees that every minimum cut 2-respects one of the trees. By our reduction, this immediately that we can deterministically compute a cactus in $O(m\,\text{polylog}(n))$ time.

{\paragraph{Application.} 
Cactus construction has an immediate application to the \emph{$+1$-edge-connectivity augmentation problem} defined as follows: given an undirected integer-weighted graph $G=(V, E, w)$ where $w: E\to \mathbb{N}$, compute an edge set $E'$ of the minimum size such that the minimum cut on $G'=(V, E\cup E')$ has value $\lambda(G)+1$, where $\lambda(G)$ is the value of the minimum cut of $G$.
To solve this problem, one can simply compute a cactus representation of $G$ and then apply in $O(n)$ time a DFS traversal algorithm from Naor, Gusfield, and Martel on top of the cactus~\cite[Section 3]{naor1997fast}.
Therefore, \Cref{thm:main} gives immediately the first $m^{1+o(1)}$-time deterministic algorithm and a $O(m\log^3 n)$-time randomized algorithms for this problem, improving the previous bound of $\Omega(m \log^4 n)$ implied by \cite{karger2009near}.\footnote{We remark that if $G$ is an unweighted graph, the $+1$-edge-connectivity augmentation problem can be solved in $O(m\log^2 n(\log\log n)^2)$ time by \cite{henzinger2020local, lo2020compact}. In addition, if the goal is to increase the edge connectivity to a particular quantity $\tau$ in a weighted graph, this variant can also be solved in $\tilde{O}(m)$ time by~\cite{CLP22b}.
}

\subsection{Related Works}

\paragraph{Deterministic Algorithms.}
For decades, significant effort has been devoted to devising deterministic algorithms that is as fast as their randomized counterparts.
Examples include the line of work on deterministic minimum edge cut algorithms \cite{kawarabayashi2018deterministic,henzinger2020local,saranurak2021simple,li2020deterministic,li2021deterministic}, deterministic minimum vertex cut algorithms \cite{gabow2006using,SaranurakY23}, and deterministic Laplacian solvers and approximate max flow \cite{chuzhoy2020deterministic}. Each of these deterministic algorithms usually deepens insight on the problems. Our result extends this line of research and reveals deeper structural understanding on 2-respecting mincuts (defined later in \Cref{sec:2respecting}).

\paragraph{Faster Randomized Global Mincut Algorithms.}
Recently, \cite{gawrychowski2019minimum} and \cite{Mukhopadhyay2020WeightedMS} gave faster randomized algorithms for computing a \emph{single} global minimum cut in an undirected weighted graph. Their algorithms run in $O(m\log^2 n)$ and $O(m\log^2 n + n\log^6 n)$ respectively, which improved Karger's long-standing $O(m\log^3 n)$ time algorithm. But their approaches are difficult to generalize to find all global mincuts. Our algorithm leaves an $O(\log n)$ gap between finding one mincut with all mincuts in randomized setting, which is left for future work.

%% file: 2-preliminary.tex
\section{Preliminaries}

Let $G=(V, E, w)$ be an undirected weighted graph, with $n$ vertices and $m$ nonnegatively weighted edges.
A cut $(X, V\setminus X)$ 
(or the cut \emph{induced by} $X$, or simply denoted by $X$ if the context is clear)
is a proper partition of $V$,
and the edges across the cut are called \emph{cut edges}.
The \emph{weight} of the cut is defined to be the sum of all edge weights across the cut, denoted by $\C(X) := \C(X, V\setminus X)$.
By extending the above notation we define for any two (not necessarily disjoint) subsets $X$ and $Y$, let $\C(X, Y)$ be the sum of weights of edges with one endpoint being in $X$ and another endpoint being in $Y$. Notice that an edge with both endpoints in $X\cap Y$ will be counted twice.

\paragraph{Global Mincuts.} The \emph{global minimum cut} or simply \emph{mincut} is a cut whose weight is the smallest among all cuts.
Throughout the paper we use $\lambda$ to denote the weight of any global minimum cut. We also assume that the value  $\lambda$ is already precomputed~\cite{li2021deterministic,karger2000minimum,gawrychowski2019minimum,Mukhopadhyay2020WeightedMS}.

\paragraph{Minimal Mincuts.} Since a cut separates some vertices from  any given vertex, throughout the paper, we will designate an arbitrary but fixed \emph{root vertex} $r\in V$. After fixing the root, we are able to characterize the mincuts that separate a vertex or an edge from the root.
The \emph{size} of a cut $(X, V\setminus X)$ where $r\notin X$ is then defined to be the number of vertices in $X$.
Intuitively, for each vertex $v$ (or an edge $e=(u, v)$), the most relevant global mincut would be the one that minimizes the number of the vertices ``on the $v$ (or $e$) side'' of the cut.
Thus, we have the following definition:

\begin{definition}[Minimal mincuts]
    The \emph{minimal mincut} of a vertex $v$ is the mincut of the least size separating $v$ from $r$. If $v$ is not separated from $r$ by any mincut, then its minimal mincut is \NULL. The minimal mincut of an edge $(u,v)$ is defined similarly except that the mincut must separate both $u$ and $v$ from $r$.
\end{definition}

\subsection{Crossing Mincuts, Uniqueness of Minimal Mincuts}

The most important property of minimum cuts along the history should be the \emph{submodularity} of crossing cuts.
Two cuts $(X, V\setminus X)$ and $(Y, V\setminus Y)$ are said to be \emph{crossing} if each of $X\cap Y, X\setminus Y, Y\setminus X, (V\setminus X)\cap (V\setminus Y)$ is non-empty. With submodularity one can show that:

\begin{lemma} [\cite{dinits1976structure}]
    \label{lem:crossing-mincuts}
    If $(X, V\setminus X)$ and $(Y, V\setminus Y)$ are crossing mincuts, then each of the cuts induced by $X\cap Y$, $X\setminus Y$, $Y\setminus X$,  and $X\cup Y$ are also mincuts. Furthermore, we have $\C(X\cap Y, (V\setminus X)\cap (V\setminus Y)) = \C(X\setminus Y, Y\setminus X) = 0$.
\end{lemma}

With the crossing property above, we are able to deduce the uniqueness of the minimal mincut of a vertex (or an edge). This uniqueness property plays an important role in many applications (closest mincut, left-most mincuts, ...etc) as well as constructing a cactus representation of mincuts.

\begin{lemma}[\cite{karger2009near}]
    \label{lem:uniqueness-minimal}
    If a minimal mincut of a vertex or edge exists, then it is unique.
\end{lemma}

\subsection{Cactus: Representation of all Global Mincuts}
Dinitz et al~\cite{dinits1976structure} showed that there exists a cactus graph $H$ (every edge belongs to at most one cycle) with $O(n)$ edges that represents all global mincuts on $G$.
The representation has the following properties.
Each vertex of $G$ is mapped to a node on $H$. (This mapping could be neither surjective or injective.)
For every edge in $H$ or two edges in the same cycle of $H$ that split all nodes in $H$ into two parts, the corresponding partition of vertices in $G$ forms a global mincut.
Conversely, for any global mincut in $G$ one can also find a corresponding edge or a pair of edges in $H$ that represents this mincut.

Although the size of the cactus representation is $O(n)$, it is highly nontrivial to compute such a representation from $G$.
Karger and Panigrahi~\cite{karger2009near} gave the first randomized Monte Carlo algorithm that computes a cactus representation in $\tilde{O}(m)$ time.  Since there can be $\Omega(n^2)$ many mincuts (in a cycle, for example), we note that this cactus has to be computed without explicitly listing all mincuts.

\section{The Framework of \cite{karger2009near} and Our Improvement}

Since our improvement is mainly based on the framework of Karger and Panigrahi~\cite{karger2009near}, in this section we introduce the framework and describe our contribution in details.

\subsection{2-Respecting Mincuts and Tree Packing}
\label{sec:2respecting}
Similar to most of the fastest exact minimum cut algorithms \cite{karger2000minimum,karger2009near,gawrychowski2019minimum,Mukhopadhyay2020WeightedMS}, our algorithm is based on computing global minimum cuts that 2-respect a spanning tree. 
Let $T$ be any spanning tree on $G$.
A cut is said to be \emph{$k$-respecting} $T$ if the spanning tree $T$ contains at most $k$ cut edges. A cut is said to \emph{strictly $k$-respect} a spanning tree of a graph if the spanning tree contains exactly $k$ cut edges.
Karger~\cite{karger2000minimum} first showed that there exists a collection $\T$ of $O(\log n)$ spanning trees in $G$ such that every global mincut 2-respects some spanning tree in $\T$.
Such collection $\T$ is also called a \emph{tree packing}.
In the same paper Karger also gave a randomized Monte Carlo algorithm that in $O(m+n\log^3n)$ time computes a tree packing with high probability.\footnote{Later on, 
Gawrychowski, Mozes and Weimann~\cite{gawrychowski2019minimum}
give another time bound $O(m\log^2 n)$, which is faster on sparse graphs, but we do not exploit this new bound in our paper.}
Recently,
Li~\cite{li2021deterministic} gave the first \emph{deterministic} algorithm that computes a tree packing of size $n^{o(1)}$ in $m^{1+o(1)}$ time.

Once a tree packing $\T$ is found, the task of searching for a global minimum cut can be reduced to checking all global minimum cuts that 2-respect a spanning tree $T\in \T$.
We summarize these useful algorithms computing tree packings in \Cref{thm:tree-packing}.

\begin{theorem}
    \label{thm:tree-packing}
    Given an undirected weighted graph $G$, there are algorithms that compute a tree packing $\cal T$ consisting of~\vspace*{-6pt}
    \begin{itemize}[itemsep=-3pt]
        \item $n^{o(1)}$ spanning trees by a deterministic algorithm in $m^{1+o(1)}$ time,
        or
        \item $O(\log n)$ spanning trees by a randomized Monte Carlo algorithm\footnote{success with high probability $1-n^{-\Theta(1)}$.} in $O(m+n\log^3 n)$ time, %
    \end{itemize}
    \vspace*{-6pt}such that each global mincut 2-respects some tree in $\T$.
\end{theorem}

Karger and Panigrahi's algorithm~\cite{karger2009near} reduces the cactus construction problem to computing the minimal mincuts of all the vertices and edges.
Besides designing an efficient algorithm which computes these minimal mincuts, 
it is also important to store these minimal mincuts in a succinct way.
Using~\Cref{thm:tree-packing}, each minimal mincut 2-respects some tree in $\T$.
As long as there is a way to represent a 2-respecting mincut on a tree using $O(\log n)$ bits, 
each minimal mincut can be represented efficiently in a total of $O(m\log n)$ bits.
We refer to these representations as \emph{cut labels}.

\subsection{Cut Labels and Three Types of 2-Repecting Mincuts}

Fix a spanning tree $T\in \T$ with root $r$.
The set of descendants of a vertex $v$ is in a spanning tree $T$ is denoted by $v_T^\downarrow$, and the set of strict descendants of $v$ is denoted by $v_T^\Downarrow$, i.e. $v_T^\Downarrow = v_T^\downarrow \setminus \{v\}$. Similarly, the set of ancestors of $v$ in $T$ is denoted by $v_T^\uparrow$, and we define $v_T^\Uparrow = v_T^\uparrow\setminus \{v\}$.
We may drop the subscript and simply denote the sets by $v^\downarrow$, $v^\Downarrow$, $v^\uparrow$, and $v^\Uparrow$ if there is no confusion. 

Consider any 2-respecting cut on $T$. This cut must intersect with at most two edges on $T$.
In particular, this cut can be either 1-respecting $T$ or strictly 2-respecting $T$.
In the case where there are two edges across the cut on $T$, these two edges may or may not be on the same path from the root $r$ to some vertex on $T$. Hence, we can classify any 2-respecting cut as one of the following three types:
\begin{description}[itemsep=0pt]
\item[Type 1.] The cut 1-respects $T$. In this case there exists a vertex $v$ so that $v_T^\downarrow$ induces the cut.
\item[Type 2-Comparable.]
The cut strictly 2-respects $T$, and the two tree edges across the cut belong to the same path from the root $r$.
Let $v$ and $w$ be the lower endpoints to the tree edges across the cut.
Then, we must have $v\in w^\downarrow$ or $w\in v^\downarrow$ and we say that $v$ and $w$ are \emph{comparable} (denoted by $v\parallel w$). Moreover, without loss of generality let $w\in v^\downarrow$, then this cut must be induced by $v_T^\downarrow\setminus w_T^\downarrow$. In this case we say that $w$ is the \emph{lower vertex} and $v$ is the \emph{upper vertex}.
\item[Type 2-Incomparable.] The cut strictly 2-respects $T$, and the two tree edges across the cut belong to different paths from $r$. Let $v$ and $w$ be the lower endpoints to the tree edges across the cut.
Then we must have $v\notin w^\downarrow$ and $w\notin v^\downarrow$. In this case we say that $v$ and $w$ are \emph{incomparable} (denoted by $v\perp w$).
Again, this cut must be induced by $v_T^\downarrow\cup w_T^\downarrow$.
\end{description}

With the above classification of 2-respecting mincuts, one immediately sees that any minimal mincut of a vertex (or an edge) can be represented by an $O(\log n)$-bit cut label of the form $(\mathit{type}, v, w, T)$. Once we obtain cut labels of minimal mincuts for all vertices and all edges, a cactus representation can be constructed efficiently:

\begin{restatable}{lemma}{cactusconstructionlemma}
\label{thm:reduction}
    Given a graph $G=(V, E)$, a tree packing $\T$ and the set of cut labels representing minimal mincuts of each vertex $v\in V$ and each edge $e\in E$, 
    there exists a deterministic algorithm that computes a cactus representation in $O(m\alpha(m, n) + n|\T|)$ time, where $\alpha(m, n)$ is the inverse Ackermann function.   
\end{restatable}
    
In \cite{karger2009near}, a similar statement was given. However, only an imprecise $\tilde{O}(m)$ time bound was given and some description of the algorithms and argument of the proofs were omitted (e.g., the proof to Lemma 3.12 in \cite{karger2009near} and the cases analysis on their last page). In this paper, we give a formal detailed proof of 
\Cref{thm:reduction}. Since the technical contribution is mainly to complete (and simplify) the argument of \cite{karger2009near} by incorporating formal concepts from \cite{gabow2016minset}, we defer the proof to \Cref{app:reduction-algorithm}.

\subsection{Computing Cut Labels Efficiently}

With the help of \Cref{thm:tree-packing} and \Cref{thm:reduction}, the task of computing a cactus representation reduces to
computing cut labels of the minimal mincuts for each vertex and each edge.
Since the number of trees in the tree packing $\T$ is small, 
it suffices to compute
a \emph{minimal 2-respecting mincut candidate} on each tree $T$ for each vertex and each edge.
That is, whenever the minimal mincut of an edge or a vertex 2-respects $T$, the returned candidate must be this mincut.

Notice that for a particular vertex $v$ (or an edge $e$), the candidate may not exist in every tree.
Now, if we are able to compute minimal 2-respecting mincut candidates in almost-linear time, we are able to compute the cactus representation in almost-linear time as well:

\begin{lemma}\label{lem:reduction-label-tree}
Suppose there is a deterministic algorithm that, given a spanning tree $T$, computes a label for every vertex $v$ and edge $e$ representing its minimal
2-respects mincut candidate in $t_{tree}$ total time.
Then, to compute the cactus representation of a graph, there is a deterministic algorithm that runs in  $t_{tree}\cdot n^{o(1)}+m^{1+o(1)}$ time, and there is also a randomized Monte Carlo algorithm that runs in 
$O(t_{tree}\log n+m\log^{2}n)$ time.
\end{lemma}

\begin{proof}
We first invoke \Cref{thm:tree-packing} and obtain a tree packing ${\cal T}$. Then, for
each spanning tree $T\in{\cal T}$, we compute a cut label for every
vertex $v$ and edge $e$ representing its minimal 2-respects mincut.
Since every mincut 2-respects one of the tree from ${\cal T}$, one
of the cut label of each vertex $v$ and edge $e$ corresponds to
its minimal mincut (and we can obtain it by comparing the weight and size
in $O(m)$ total time). In total, this takes either $t_{tree}\cdot n^{o(1)}+m^{1+o(1)}$
time deterministically or $O(t_{tree}\log n+m\log^{2}n)$ time Monte-Carlo
randomized. 

Given labels representing minimal mincuts of all vertices
and edges, we obtain a cactus representation in $O(m\log^{2}n)$ additional
time by \Cref{thm:reduction}.
\end{proof}

\paragraph{Computing Minimal 2-Respecting Mincut Candidates.} 
The remaining task is to compute cut labels of the minimal 2-respecting mincut candidates on every tree $T\in \T$, for each vertex and each edge.
Karger and Panigrahi~\cite{karger2009near} provided a deterministic algorithm (see~\Cref{lem:labeling-vertices}) for computing these cut labels for vertices --- they partially bypassed the challenge for computing cut labels of the edges via randomization. 
However, there is a missing case in~\cite{karger2009near} when computing cut labels for vertices.

Regarding the gap in~\cite{karger2009near}, we believe that the approach is not wrong, but the fix seems to require more than changing some typos. 
The following \Cref{lem:labeling-vertices,lem:labeling-edges} summarizes the tasks for obtaining cut labels to minimal 2-respecting mincut candidates for vertices and edges on a given tree $T$.

\begin{lemma}[\cite{karger2009near}]%
    \label{lem:labeling-vertices}
    Given a spanning tree $T$, we can  deterministically compute, for each vertex $v$, a cut label representing a minimal 2-respecting mincut candidate of $v$ in $O(m\log^2 n)$ total time.
\end{lemma}

In this paper, we give a simpler and complete algorithm of \Cref{lem:labeling-vertices} in \Cref{sec:labeling-vertices}.

\begin{lemma}[Key Lemma]
    \label{lem:labeling-edges}
    Given a spanning tree $T$, we can  deterministically compute, for each edge $e$, a cut label representing a minimal 2-respecting mincut candidate of $e$ in $O(m\log^2 n)$ total time.
\end{lemma}

The algorithm and the proof to~\Cref{lem:labeling-vertices} will be in \Cref{sec:labeling-vertices}.
The rest of our paper is devoted to proving \Cref{lem:labeling-edges}.
By plugging \Cref{lem:labeling-vertices,lem:labeling-edges} into \Cref{lem:reduction-label-tree}, we can conclude \Cref{thm:main}. 

\begin{proof}[Proof of~\Cref{thm:main}]
By \Cref{lem:labeling-vertices,lem:labeling-edges}, we have that $t_{tree}=O(m\log^2 n)$.
Thus, by \Cref{lem:reduction-label-tree} we obtain a deterministic algorithm computing a cactus respresentation in $m^{1+o(1)}$ time, and also a randomized Monte Carlo algorithm that runs in $O(m\log^3 n)$ time.
\end{proof}

\paragraph{Toward the Proof of \Cref{lem:labeling-edges}.}
We show a path to prove our key lemma. The goal is to find a minimal 2-respecting mincut candidate for each edge $e$ on a given tree $T$.
Since there are three types of cuts that 2-respects a tree $T$, it is natural to split the task into three subproblems, with each of them focusing on Type 1 cuts, Type 2-Comparable cuts, and Type 2-Incomparable cuts respectively.

\paragraph{Minimal 1-respecting mincut candidate for vertices and edges} Computing the minimal 1-respecting mincut candidate for all vertices and edges is relatively simple, and can be derived from~\cite{karger2000minimum}. For the sake of completeness we include the proof below.

The following basic lemma from~\cite{karger2000minimum} can be implemented by dynamic program.

\begin{lemma}[Lemma 5.1 in \cite{karger2000minimum}]
\label{lem:computing-all-1-respecting}
The values of all cuts that 1-respect a given spanning tree $T$ can be determined in $O(m+n)$ time.
\end{lemma}

In the second step, we show how to compute the minimal 1-respect mincut candidate for vertices.

\begin{lemma}\label{lem:labeling-vertex-1-respecting}
    Given a graph $G$ and a spanning tree $T$,
    there is an algorithm such that, in $O(m+n)$ time computes
    the minimal 1-respecting mincut candidates for all vertices.
\end{lemma}

\begin{proof}[Proof of \Cref{lem:labeling-vertex-1-respecting}.]
First, we invoke \Cref{lem:computing-all-1-respecting} to compute the value of all the 1-respecting cut. By comparing the value of each 1-respecting cut with the value of mincut $\lambda$, we can identify all the 1-respecting mincuts. If the 1-respecting cut $v^\downarrow$ seperating $u$ from $r$, then $u$ must in the subtree of $v$. Therefore, we can run a DFS and maintain the minimal 1-respecting mincut containing the current vertex $u$, which can be done in linear time.
\end{proof}

Finally, we get the minimal 1-respecting mincut candidate for all edges.
\begin{restatable}[]{lemma}{lemOneRespecting}\label{lem:1-respecting}
    Given a graph $G$ and a spanning tree $T$,
    there is an algorithm such that, in $O(m+n)$ time computes
    the minimal 1-respecting mincut candidates for all edges.
\end{restatable}

\begin{proof}
Observe that a 1-respecting cut contains $e = (u_1,u_2)$ if and only if it contains $\lca_e = \LCA(u_1,u_2)$. By \Cref{lem:labeling-vertex-1-respecting}, we compute the minimal 1-respecting mincut for all the vertices in $O(m+n)$ time. In addition, computing $\lca_e$ for all the edges can be done in linear time. Therefore, we get the minimal 1-respecting mincut candidate for all the edges in $O(m+n)$ time.
\end{proof}

Therefore, the remaining challenges are finding
a strictly comparable 2-respecting mincut candidate (see \Cref{lem:labeling-edges-comparable}) and a strictly incomparable 2-respecting mincut candidate (see \Cref{lem:labeling-edges-incomparable}) for each edge.

\subsection{Technical Contribution}
\label{sec:comparison-with-KP09}

\paragraph{New Algorithm for Minimal 2-Respecting Mincuts.}
\Cref{lem:labeling-edges} is the key technical contribution. We give the first deterministic algorithm for computing minimal mincuts for edges in almost-linear time, which leads to the deterministic algorithm for computing cactus representation.
With \Cref{lem:labeling-edges}, we also obtain an $O(m\log^3n)$ randomized algorithm for computing a cactus representation, while the previous fastest (randomized) algorithm by  \cite{karger2009near} requires $\Omega(m\log^4 n)$ runtime (see \Cref{sec:kp-needs-more-time}.)

In fact, our algorithm is more modular than the algorithm by Karger and Panigrahi~\cite{karger2009near} in the following sense: their algorithm only computes minimal mincuts for only \emph{some} random edges, but they show that this set of edges is sufficient. It requires more intricate proof to show that these mincuts suffice for constructing a correct cactus representation. The approach makes the overall framework less modular.

\paragraph{Structural Properties for 2-Respecting Mincuts.}
What enables us to achieve \Cref{lem:labeling-edges} are new structural lemmas about 2-respecting mincuts. 
2-respecting mincuts are not esoteric objects.
In fact, fast algorithms related to 2-respecting mincuts have been the only known pathway for obtaining global mincuts in general weighted graphs in near-optimal complexity in many models of computation (sequential \cite{karger2000minimum}, parallel \cite{GeissmannG21},
distributed \cite{DoryEMN21},
streaming and cut queries \cite{Mukhopadhyay2020WeightedMS}.
We are hopeful that these structural lemmas are promising for these models too. 

For example, when the minimal mincut of an edge $e$ is a comparable 2-respecting cut,
a non-trivial observation is that the task is reduced to computing only the \emph{lower vertex} $l_e$ of the two vertices that define the 2-respecting cut.
This allows us to focus on just ``half of the problem'' and hence the algorithm can be greatly simplified.
Structural insights of \Cref{lem:uniqueness-of-lower-vertex}, \Cref{cor:uniqueness-of-lower-vertices}, and \Cref{lem:general-case-property} allow us to design a DFS procedure that obtains these lower vertices in \emph{just one pass}. 
\Cref{lem:invariant-of-w} plays a similarly important role when the minimal mincut of an edge $e$ is an incomparable 2-respecting cut.

\paragraph{A Full Detailed Proof to Cactus Construction.}
We also provide a not only comprehensive but also simplified algorithm in \Cref{app:reduction-algorithm} for constructing a cactus from the labels of minimal mincuts of vertices and edges.
Some correctness proofs in the last section of Karger and Panigrahi's paper~\cite{karger2009near} were missing, and unfortunately there was no full version.
By revisiting the work of Gabow's~\cite{gabow2016minset} and Karger and Panigrahi's~\cite{karger2009near}, we believe \Cref{app:reduction-algorithm} helps the readers and the community understand Karger and Panigrahi's algorithm with a  much higher confidence.

\paragraph{Other Technical Contributions.}
Besides the key technical contribution and the simplified cactus construction algorithm,
we also introduce a new algorithm for computing minimal mincut of vertices, which is simpler and fixes a gap in \cite{karger2009near}. This algorithm serves as an alternative proof to~\Cref{lem:labeling-vertices} and is described in \Cref{sec:labeling-vertices}.
Last but not least, we formalize the reduction to path using path decomposition in \Cref{lem:reduc to path}, which allows us to focus on paths. This reduction simplifies the description of the algorithm, makes the analysis more modular, and can become handy in other applications.

\section{Useful Tools}\label{sec:useful-tools}

\paragraph{Reduction to Paths via Path Decomposition.}
To facilitate our algorithm throughout the paper, we utilize the reduction that reduces
a problem on a tree to several problems on a collection of paths.
Similar techniques have been developed in order to solve problems related to minimum cuts that 2-respecting trees. 
Two specific ways of decomposing a tree into paths were used: \emph{bough decomposition}~\cite{karger2000minimum} and \emph{heavy path decomposition}~\cite{gawrychowski2019minimum, BhardwajLovettSandlund2020} (see also \cite{Mukhopadhyay2020WeightedMS, GeissmannG21}
for more discussions).

It turns out that all we need is a balanced property for any decomposition of a tree into a collection of paths.
Let $T$ be a tree rooted at $r$. An \emph{oriented path} on $T$ is a path $P$ with the vertex closest to $r$ being an endpoint.
A \emph{path decomposition} $\mathcal{P}$ of $T$
is a collection of oriented paths on $T$ so that each vertex of $T$ belongs to exactly one path.
We say that a path decomposition $\cP$ is \emph{balanced} if for any vertex $v$,
the path from the root of the tree to $v$ intersects with $O(\log n)$ paths in $\mathcal{P}$.
Both bough decomposition and heavy path decomposition of a tree are balanced, and can be computed in linear time.

With the balanced property, it is straightforward to see that there will be only an $O(\log n)$ overhead if we are allowed to process each path $P\in\cP$ with a runtime related to the size of the subtree rooted at the highest vertex of $P$ (e.g., perform a DFS).
Specifically, given a path $P$, we define $P^\downarrow$ to be the set of all vertices with at least one ancestor in $P$.
Let $E(P^\downarrow)$ be the set of edges incident to at least one vertex in $P^\downarrow$ and let $d(P^\downarrow)=\sum_{v\in P^{\downarrow}} \deg(v)$ be the \emph{unweighted volume} of the subtree.
\Cref{lem:reduc to path} below describes how we will bound the total running time using a balanced path decomposition in this paper.

\begin{lemma}
\label{lem:reduc to path}
If there exists an algorithm that preprocess $G$ and a spanning tree $T$ in $t_p$ time such that, for any path $P$ in a balanced path decomposition $\mathcal{P}$ of $T$, and a specific function $g(e,P)$, computes $g(e, P)$ for all $e\in E(P^\downarrow)$ in total time $O(d(P^\downarrow)\log n)$.
Then, we can compute in $t_p + O(m\log^2 n)$ time  $g(e, P)$ for all $e\in E$ and for all $P$ where $e \in E(P^\downarrow)$.

\end{lemma}

\begin{proof}
By the property of balanced path decomposition, for each edge $e$, there are at most $O(\log n)$ paths $P$ such that $e\in E(P^\downarrow)$. Therefore, we have
\begin{equation}
\label{eqn:sum of paths}
    \sum_{P^\in \cP} d(P^\downarrow) = O(m\log n) ~.
\end{equation}
For each $P\in \cP$, since the algorithm computes $g(e,P)$ values every $e\in E(P^\downarrow)$ in $O(d(P^\downarrow))$ time. We can compute $g(e,P)$ for all $e\in E$ and all $P$ where $e\in E(P^\downarrow)$ in total time $O(\sum_{P^\in \cP} d(P^\downarrow)\log n) = O(m\log^2 n)$. Adding the preprocessing time $t_p$, the algorithm runs in $t_p+O(m\log^2 n)$ time.
\end{proof}

The usages of \Cref{lem:reduc to path} in this paper are quite similar in the taste: suppose we would like to compute some information (e.g., a minimal incomparable 2-respecting mincut candidate) of an edge $g(e)$, and realizes that $g(e)$ can be computed efficiently from the set $\{g(e, P)\}$ where $e\in P^\downarrow$ (e.g., a minimal 2-respecting mincut candidate with one crossing edge on $P$). Then by applying~\Cref{lem:reduc to path} we can focus on computing $g(e, P)$ values for each specific path $P\in\cP$.
We apply~\Cref{lem:reduc to path} in many cases in \Cref{sec:label-comparable}, \Cref{sec:label-incomparable}, and \Cref{sec:labeling-vertices}. 

The path decomposition $\cP$ we use throughout in this paper will be assumed to be balanced.

\paragraph{Data Structures on Trees.}
The second tool that are extensively used are dynamic tree data structures (and top-tree data structures).
These data structures maintain values associated with vertices, with the list of operations supported in~\Cref{lem:data structures}.

\begin{lemma}
\label{lem:data structures}
There exists a data structure over a dynamic forest of $n$ vertices, supporting the following operations in the worst case $O(\log n)$ time:
\textup{
\begin{itemize}[itemsep=0pt]
    \item $\Link(v,w)$: where $v, w$ are in different trees, links these trees by adding the edge $(v,w)$ to our dynamic forest.
    \item $\Cut(e)$: remove edge $e$ from our dynamic forest.
    \item $\AddP(u,x)$: add $x$ to the value of every vertices on the path from $u$ to the root.
    \item $\MinP^\downarrow(u)$: return argmin of the value of vertex on the path from $u$ to the root, and break tie by finding the deepest one.
    \item $\MinP^\uparrow(u)$: the same as $\MinP^\downarrow(u)$, but break tie by finding the highest one.
\end{itemize}
}
All these operations can be supported with a dynamic tree \cite{sleator1983data}. Besides, we need the following operations, which can be implemented using top-tree \cite{alstrup2005maintaining}.
\textup{\begin{itemize}
    \item $\MinTreeDown(u)$: returns a vertex $v$ with minimum value in the tree $T$ that contains $u$, breaking tie by finding the one with the smallest subtree size $|v^\downarrow|$. %
    \item $\MinTreeUp(u)$: the same as $\MinTreeDown(u)$, %
    but break tie by finding the one with the largest subtree size $|v^\downarrow|$.
    \item $\MinNonPath(v,w)$%
    : where $v, w$ are in the same tree $T$, return the vertex $u$ with the minimum value such that $u\in T$ but $u$ is not in the path between $v$ and $w$.
\end{itemize}}
\end{lemma}

\begin{proof}
    The operations $\Link(v,w)$, $\Cut(e)$ and $\AddP(u,x)$ are basic primitives of dynamic tree \cite{sleator1983data}. To implement $\MinP^\uparrow(u)$ and $\MinP^\downarrow(u)$, we just need to show how to break tie, since finding the argmin along the path is also a primitive. WLOG we consider implementing $\MinP^\uparrow(u)$. In the preprocessing step, we add  $-\epsilon|u^\downarrow|$ to the value of vertex $u$ where $\epsilon \ll 1/n$. Then for the values of two vertices equals before, they will become difference since the two vertices are comparable for they are on the path from $u$ to the root, and the value of the higher one will become smaller. By reversing the sign the same approach works for $\MinP^\downarrow(u)$.
    
    Since $\MinTree(u)$ is a primitive of top-tree\footnote{Theorem~4 in \cite{alstrup2005maintaining}.}, we can use the same approach as above to break tie. Finally $\MinNonPath(v,w)$ can be implemented by the \MaxNonPath{} primitive of a top-tree, which appears in the proof of Theorem~4 in \cite{alstrup2005maintaining}.
\end{proof}

%% file: 3-comparable.tex
\section{Comparable 2-respecting Minimal Mincuts of Edges}

\label{sec:label-comparable}

In this section, we present the algorithm computing the minimal mincut of edge when it is a comparable 2-respecting mincut of $T$. Henceforth, for every edge $e\in E$ we call $T$ the \emph{right} tree for $e$ if the minimal mincut of $e$ is a comparable 2-respecting mincut of $T$. (In this case we also call $e$ a \emph{right edge} in $T$.)
In particular, we can represent this minimal mincut using a vertex pair $(u_e, l_e)$ on the tree, indicating that $u_e^\downarrow\setminus l_e^\downarrow$ is the comparable 2-respecting mincut we found for the edge $e$.
For convenience, for any comparable 2-respecting mincut $w^\downarrow\setminus v^\downarrow$, we call $v$ the \emph{lower vertex} and $w$ the \emph{upper vertex} in the cut $w^\downarrow\setminus v^\downarrow$.

\Cref{lem:labeling-edges-comparable} summarizes the algorithm that computes such a vertex pair $(u_e, l_e)$ for every edge $e$ on any given spanning tree $T\in \T$ in $O(m\log^2 n)$ time.
Notice that when we analyze the correctness of the algorithm on a spanning tree $T$, it suffices to focus on the edges where $T$ is the right tree.
In the case that $T$ is \emph{not} the right tree for an edge $e$, it could be that the returned vertex pair $(u_e, l_e)$ be either some arbitrary mincut or it could be $(\NULL, \NULL)$.

\begin{lemma}
    \label{lem:labeling-edges-comparable}
    There is an algorithm that, given a graph $G=(V,E)$ and a spanning tree $T\in \T$, in total time $O(m\log^2 n)$ computes, for every edge $e \in E$, a vertex pair $(u_e,l_e)$ where $u_e,l_e \in V \cup \{ \NULL \}$
    with the following guarantee: if $T$ is the right tree for $e$, then
    $u_e^\downarrow \setminus l_e^\downarrow$ is the minimal mincut of $e$.
\end{lemma}

Our algorithm is divided into two main steps. In the first step the lower veritces $l_e$ are computed. Then based on $l_e$, in the second step the algorithm finds their correponding upper vertices $u_e$.

It turns out that the second step becomes simpler once we have computed the lower vertices $l_e$ for all edge $e$.
This reduction is presented in \Cref{sec:reduc to lower}.
Surprisingly, there is a deterministic algorithm that guarantees to find lower vertices $l_e$ for each right edge $e$ efficiently. We present the most important structural property supplemented with the algorithm in \Cref{sec:compute lower vertex}.

\subsection{Reduction to Computing Lower Vertices}
\label{sec:reduc to lower}

Fix a spanning tree $T\in \T$ and suppose that we have already obtained all lower vertices $l_e\in V\cup \{\NULL\}$ for each edge.
\Cref{thm:reduction-to-lower-vertices} states that there exists an efficient algorithm that guarantees to find corresponding upper vertices $u_e$ for all right edges on $T$.

\begin{lemma}[Reduction to lower vertices]
\label{thm:reduction-to-lower-vertices}
There is an algorithm that, given a graph $G=(V,E)$, a spanning tree
$T$, and a lower vertex $l_{e}\in V \cup \{\NULL\}$ for every edge $e \in E$, in time $O(m\log^2 n)$ computes an upper vertex $u_{e}\in V\cup \{\NULL\}$
for every edge $e\in E$ with the following guarantee: 
If $T$ is the right tree for $e$ and $l_{e}$ is the lower vertex of the minimal mincut of edge $e$, then $u_{e}$ is the upper vertex of the minimal mincut.
\end{lemma}

Intuitively, for any lower vertex $l_e$ of $e=(u_1, u_2)$, if we obtain the list of \emph{comparable partners} who form comparable 2-respecting mincuts with $l_e$, then the upper vertex we are looking for must be the lowest ancestor of $\LCA(l_e, \LCA(u_1, u_2))$ which appears in the list.
However, obtaining the list is inefficient.
Fortunately, based on the path decomposition,
we can maintain all such candidates of upper vertices on-the-fly and answer all the queries using a dynamic data structure. Specifically,  
whenever the algorithm processes a vertex $v$ on a path $P\in\cP$, this data structure finds all upper vertices for all edges in $Q_v=\{e\in E \ |\ l_e=v\}$, which is 
summarized in \Cref{lem:labeling-upper-vertex}.
Hence, by applying a very similar path decomposition framework as in \Cref{lem:reduc to path}, we can compute all upper vertices efficiently, which is summarized in \Cref{lem:amortized-path-decomposition}.

\begin{lemma}
\label{lem:labeling-upper-vertex}
    Let $\cP$ be a path decomposition of $T$. We can preprocess the graph $G$ and the spanning tree $T$ in $O(m\log n)$ time so that, given any path $P \in \cP$, we can compute the upper vertex $u_e$ for \Cref{thm:reduction-to-lower-vertices} for every edge $e\in \bigcup_{v\in P} Q_v$ in $O(d(P^\downarrow)\log n + \sum_{v\in P}|Q_v|\log n)$ time.
\end{lemma}

\begin{lemma}[A variant of path decomposition]
\label{lem:amortized-path-decomposition}
 Let $g$ be a function of $e\in E$.
    If there exists an algorithm that preprocess $G$ and a spanning tree $T$ in $t_p$ time such that, for any path $P$ in a balanced path decomposition $\mathcal{P}$ of $T$, any partition of subset of $E'\subseteq E$ into $\{Q_v\}_{v\in V}$, computes $g(e)$ for all $e\in \bigcup_{v\in P} Q_v$
    in total time $O((d(P^\downarrow) + \sum_{v\in P}|Q_v|)\log n)$, then we can compute in $t_p + O(m\log^2 n)$ time $g(e)$ for all $e\in E'$.
\end{lemma}

\begin{proof}
For each $P\in \cP$, the algorithm computes $g(e)$ values for all $e\in \bigcup_{v\in P} Q_v$
in total time $O((d(P^\downarrow) + \sum_{v\in P}|Q_v|)\log n)$. Since $\bigcup_{v\in V} Q_v = E'$ and $\bigcup_{P\in \cP}P = V$, we can compute $g(e)$ for all $e\in E'$ by running the algorithm on every path $P\in \cP$. Summing up over $P\in \cP$, we can compute $g(e)$ for all $e\in E'$ in total time $O(\sum_{P^\in \cP} (d(P^\downarrow) + \sum_{v\in P}|Q_v|)\log n) = O(m\log^2 n + m\log n) = O(m\log^2 n)$ by \Cref{eqn:sum of paths}. Adding the preprocessing time $t_p$, the algorithm runs in $t_p+O(m\log^2 n)$ time.
\end{proof}

By plugging in \Cref{lem:labeling-upper-vertex} to \Cref{lem:amortized-path-decomposition}, we obtain \Cref{thm:reduction-to-lower-vertices}.

\begin{proof}[Proof of \Cref{thm:reduction-to-lower-vertices}]
First, for each edge $e$ with $l_e=\NULL$, we set $u_e=\NULL$.
    Then, we partition the set of the remaining edges into $\{Q_v\}_{v\in V}$ such that $Q_v$ contains all the edges $e$ with the  lower vertex $l_e = v$. 
    By \Cref{lem:labeling-upper-vertex}, there exists an algorithm that preprocess $G$ and tree $T$ in $O(m\log n)$ time such that, given a path $P$ in a path decomposition ${\cal P}$, compute $u_e$ for every edge $e\in \bigcup_{v\in P} Q_v$ in total time $O((d(P^\downarrow)+\sum_{v\in P}|Q_v|)\log n)$.
    Since it satisfies the condition of \Cref{lem:amortized-path-decomposition}, we can compute the upper vertex $u_e$ for every edge in $O(m\log^2 n)$ time.
\end{proof}

Now we describe an algorithm that achieves \Cref{lem:labeling-upper-vertex}.

\begin{proof}[Proof of \Cref{lem:labeling-upper-vertex}]

Fix a spanning tree $T\in \T$.
For any vertex $v\in V$ and its ancestor $w\in v^\Uparrow$, the weight of the comparable cut $w^\downarrow\setminus v^\downarrow$ is given by:
\begin{equation}\label{eqn:comparable-2-respecting-mincut}
    \C(w^\downarrow \setminus v^\downarrow) = \C(w^\downarrow) - \C(v^\downarrow) + 2(\C(v^\downarrow,w^\downarrow) - \C(v^\downarrow,v^\downarrow)) ~.
\end{equation}

Suppose $v$ is the lower vertex of the minimal mincut of $e$. %
Then by factoring out the terms only related to $v$, it suffices to compute the following \emph{comparable precut values} for all $w\in v^\Uparrow$.

\begin{definition}[Comparable precut value]
    The \emph{comparable precut value} of $v$ at $w$, is defined by 
    \[
        \CC_v(w) := \C(w^\downarrow) + 2\C(v^\downarrow,w^\downarrow) ~.
    \]
    Note that the value is only defined for $w\in v^\Uparrow$.
    We say that $w$ is a \emph{comparable partner} or just a partner of $v$ if $w$ is a minimizer of comparable precut value at $v$.
\end{definition}

To see the high level idea, we first show how to compute the upper vertex $u_e$ for every edge $e\in Q_v$ assuming that the comparable precut value of $v$ at $w$ has already been computed for every $w\in v^\Uparrow$.
Suppose $e=(u_1, u_2)\in Q_v$ is a right edge on $T$. Let $\lca_e$ denote $\LCA(u_1,u_2)$.
The minimal mincut of $e$ must contain the vertex $x_e = \LCA(v, \lca_e)$ since the minimal mincut is a comparable 2-respecting mincut in $T$ with the lower vertex $l_e=v$. 
As the minimal mincut $u_e^\downarrow\setminus l_e^\downarrow$ of edge $e$ is the mincut satisfying the condition above with the minimal size, we have that $u_e$ must be the lowest partner of $v$ such that $u_e\in x_e^\uparrow$, which can be found using $\MinPDown(x_e)$ (see \Cref{fig:ue-by-minpath}). 
Below we show how to remove the assumption that comparable precut values of $v$ at all $w\in v^\Uparrow$ have been precomputed.%

\begin{figure}[h]
\centering
\includegraphics[]{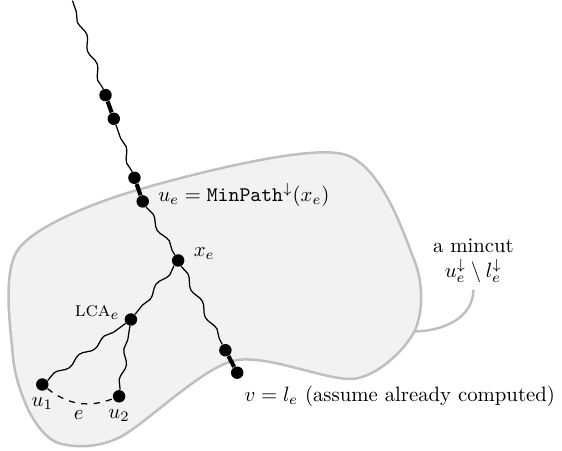}
\caption{High level idea: once the lower vertex of an edge $l_e$ has found, we may locate the upper vertex $u_e$ along the path via a $\MinPDown$ query.}\label{fig:ue-by-minpath}
\end{figure}

Given a path $P=(v_1, v_2, \ldots, v_k)$ from the path decomposition with $v_1$ being the deepest vertex, our algorithm will process $v_i$  starting from $i=1, 2, \ldots, k$.
We will maintain the invariant that once we process the vertex $v_i$ the values $\CC_{v_i}(w)$ for all $w \in v_i^\Uparrow$ can be accessed via $val[w]$.

Next we show how to maintain the invariant. 
In the preprocessing step before the path $P$ was given, we set $val[w] = \C(w^\downarrow)$ for each vertex $w$ and create a dynamic tree on $T$ (\Cref{lem:data structures}), which can be done in $O(m\log n)$ time.\footnote{\label{note:preprocess}We first compute $\lca_e=\LCA(u_1, u_2)$ for each edge $e=(u_1, u_2)$ in linear time~\cite{gabow1985linear}. Then it suffices for obtaining $val[w]=\C(w^\downarrow)$ by invoking $\AddP(u_1, w(e))$, $\AddP(u_2, w(e))$, and $\AddP(\lca_e, -2w(e))$ for every edge $e$ in $O(m\log n)$ total time.}
Now we start from the deepest vertex $v_1$, the algorithm needs to add $2\C(v_1^\downarrow, w^\downarrow)$ to each $val[w]$ so that $val[w] = \CC_{v_1}(w)$.
For each edge $(u,u')$ where $u\in v_1^\downarrow$,
we invoke $\AddP(u', 2\C(u,u'))$ so that 
two times the weight of the edge $(u,u')$ is added to $val[w]$ for each $w\in u'^\uparrow$.
The total time to recover the invariant is $O(d(v_1^\downarrow)\log n)$.

After obtaining $\CC_{v_1}(w)$ values, for every $e\in Q_{v_1}$,
we first compute $x_e = \LCA(v_1,\lca_e)$ and then compute $\tilde{u}_e = \MinPDown(x_e)$.
Finally, we check if $\tilde{u}_e$ is a partner of $v_1$. If so, then we know that the upper vertex $u_e = \tilde{u}_e$ because $\MinPDown(x_e)$ the lowest ancestor of $x_e$ that is a partner of $v_1$ if there exists some partners.
For running time, by \Cref{lem:data structures}, this requires $O(\log n)$ time for each $e\in Q_{v_1}$. So the total time for finding the upper vertices $u_e$ such that their corresponding lower vertex is $l_e = v_1$ is $O(|Q_{v_1}|\log n)$ time.
Therefore, the whole process on vertex $v_1$ can be done in $O(d(v_1^\downarrow)\log n + |Q_{v_1}|\log n)$ time.

Then, the algorithm scans through the rest of vertices $v_2,v_3,\cdots v_k$ on the path one by one. Suppose the algorithm reaches $v_i$ now. The algorithm is similar to what we did at $v_1$. 
With the invariant after processing $v_{i-1}$,
it suffices to add $2\C(v_i^\downarrow\setminus v_{i-1}^\downarrow, w)$ to $val[w]$ for each $w\in v_i^\uparrow$ by invoking $\AddP(u', 2\C(u, u'))$ for each edge $(u, u')$ where $u\in v_i^\downarrow\setminus v_{i-1}^\downarrow$. (These edges can be found in $O(d(v_i^{\downarrow}\setminus v_{i-1}^\downarrow))$ time using a DFS from $v_i$ without searching the subtree rooted at $v_{i-1}$.) 
Therefore, in $O(d(v_i^\downarrow\setminus v_{i-1}^\downarrow)\log n)$ time,  $val[w]$ are updated to $\CC_{v_i}(w)$ for all $w\in v_i^\Uparrow$, and then the algorithm uses $\MinPDown(x_e)$ to find $u_e$ for all $e\in Q_{v_i}$ in $O(|Q_{v_i}|\log n)$ time.

After finishing all the process on the path $P$, we need to roll back to the initial state $val[w] = \C(w^\downarrow)$ in order to process other paths. Therefore, the algorithm computing $u_e$ for every $e\in \bigcup_{v\in P} Q_v$ takes $O(d(P^\downarrow)\log n + \sum_{v\in P} |Q_v|\log n)$ time in total, which proves \Cref{lem:labeling-upper-vertex}.
\end{proof}

\subsection{Computing the Lower Vertex for each Edge}
\label{sec:compute lower vertex}

In the rest of this section, we show how to compute the lower vertex for each edge.
Fix a tree $T\in \T$. If the minimal mincut of edge $e$ is a comparable 2-respecting mincut of $T$, then the algorithm will find the lower vertex $l_e$ for $e$. \Cref{lem:computing-lower-vertex} summarizes the algorithm:

\begin{lemma}
\label{lem:computing-lower-vertex}
    There is an algorithm that, given a graph $G = (V,E)$ and a spanning tree $T$, in $O(m\log^2 n)$ time computes a vertex $\hat{l}_e$ for every edge $e$, such that if $e$ is a right edge in $T$ where the minimal mincut of $e$ is $u_e^\downarrow\setminus l_e^\downarrow$, then $\hat{l}_e = l_e$.
\end{lemma}

The algorithm described in \Cref{lem:computing-lower-vertex} consists of two parts.
In the first part the algorithm computes \emph{the highest partner} $H(v)$ (defined below) for each vertex $v\in V$ as a preprocessing step.
Then in the second part we apply a specialized depth first search that obtains $l_e$ values for all $e\in E$.

\subsubsection{Computing Highest Partner of each Vertex}
\label{sec:computing-highest-partner}

Fix a spanning tree $T\in \T$. For each vertex $v$, we denote $H_T(v)$ the highest (comparable) partner of $v$.
If there is no comparable 2-respecting mincut with lower vertex $v$, then $H_T(v):=\NULL$.
When there is no confusion, we shall drop the subscript $T$ and simply denote it by $H(v)$.
The goal for the algorithm is to compute $H(v)$ for all $v\in V$.

We will use the reduction to path from \Cref{lem:reduc to path}.
For any $P\in \cP$, define $g(e, P)=H(v)$ if $v\in P$ and $e$ is the tree edge with $v$ being the lower vertex, otherwise $g(e, P)=\NULL$.

Given a path $P=(v_1, v_2, \ldots, v_k)\in \cP$ with $v_1$ being the deepest vertex, the algorithm computes the highest partner for each $v_i\in P$ as follows.
Similar to the proof of \Cref{lem:labeling-upper-vertex}, the algorithm processes the vertices in the order $v_1,v_2,\ldots,v_k$.
A dynamic tree on $T$ is used and $val[w]$ is maintained such that after processing $v_i$ we obtain precut values $val[w]=\CC_{v_i}(w)$ for all $w\in v_i^\uparrow$.
Then, using a dynamic tree query $\MinPUp(v_i)$ the algorithm obtains a highest vertex $w$ with the minimum precut value $\CC_{v_i}(w)$.
Finally, we are able to assign
$g(e, P)=w$ (where $e$ is the tree edge with $v_i$ being the lower vertex, i.e., $H(v_i)=w$)
if the cut $w^\downarrow\setminus v^\downarrow$ is indeed a mincut. According to \Cref{eqn:comparable-2-respecting-mincut}, checking whether $\lambda = \C(w^\downarrow\setminus v^\downarrow)$ can be done in constant time as long as the value $\C(v^\downarrow) + 2\C(v^\downarrow, v^\downarrow)$ is precomputed.

From the discussion above, we have an algorithm that, given a path $P$, computes $g(e, P)$ for all $e \in E(P^\downarrow)$ in $O(d(P^\downarrow)\log n)$ time. The preprocessing step is the same with the one in \Cref{lem:labeling-upper-vertex}, which can be done in $O(m\log n)$ time.
By plugging in the path decomposition \Cref{lem:reduc to path}, 
we obtain an algorithm that computes $H(v)$ of all vertices $v\in V$ in $O(m\log^2 n)$ total time.

%% file: 3-main-new.tex
\subsubsection{Main Algorithm for Computing Lower Vertices}

\label{sec:compute lower main}

In this subsection we state the main algorithm for \Cref{lem:computing-lower-vertex}.
Recall that the goal is, for every right edge $e$ in $T$, we want
to compute its lower vertex $l_{e}$. Recall that when we say that
$l_{e}$ is a lower vertex, we means that there exists some
$u_{e}$ where $u_{e}^{\downarrow}\setminus l_{e}^{\downarrow}$ forms
a minimal mincut of some right edge $e$. For convenience, we denote
$\upath vw$ as the set of vertices on the path between any two vertices
$v$ and $w$ on $T$. For any right edge $e$, let $\hatP_{e}=\upath{u_{e}}{l_{e}}$
be its \emph{canonical path}. For any $e=(u_{1},u_{2})$, let $\lca_{e}=\LCA(u_{1},u_{2})$.

\paragraph{Motivation: high-level approach and the key structural lemma.}

At the highest level, the description of our algorithm is as follows.
We will perform a post-order traversal on the tree $T$ (i.e. if $a$
is an ancestor of $b$ then $a$ is visited after $b$). When we visit
$u$, we will be able to compute \emph{some} lower vertices $l_{e}$
of edges $e$ where $l_{e}$ is below $u$ (i.e.~$l_{e}\in u^{\Downarrow}$).
At the end, we make sure that we have computed all lower vertices
$l_{e}$ of all the right edges $e$ in $T$. To specify our algorithm
in more details, we start with this definition. 
\begin{defn}
[Valid lower vertices below $u$]We call a vertex $v$ a \emph{valid
lower vertex below $u$ }if $v\in u^{\Downarrow}$ and $\Maximal(v)\in u^{\uparrow}$. Let $L_{u}$ denote the set of
all valid lower vertices below $u$.
\end{defn}
In other words, $v\in L_{u}$ if there exists a comparable 2-respecting
mincut containing $u$ with the lower vertex $v$. Since $H(l_{e})$
must be an ancestor of $u_{e}$, we have following:
\begin{fact}
For every right edge $e$, $l_{e}\in L_{u}$ for any $u\in\hatP_{e}\setminus\{l_{e}\}$.
\end{fact}

Suppose that, magically, there is a data structure that, given a vertex
$u$, can return the set $L_{u}$ of all valid lower vertices below
$u$. One idea would be that whenever the post-order traversal visits
$u$, we query the data structure with $u$. Then, whenever $u\in\hatP_{e}\setminus\{l_{e}\}$,
then $l_{e}$ would be reported. However, there is an obvious issue
in this approach: the total size of $L_{u}$ over all $u$
is simply too large to be reported quickly. 

Therefore, we should consider a small subset of $L_{u}$ that still
contains $l_{e}$. Which subset of $L_{u}$ satisfies this? Intuitively, since  $u_{e}^{\downarrow}\setminus l_{e}^{\downarrow}$
is a minimal mincut,
$l_{e}$ should be ``as high as possible'' (and $u_{e}$ should
be ``as low as possible''). This motivates the following definition: for
any vertex set $S$, the set of \emph{top vertices}
of $S$, denoted by $\highest(S)$, contains all vertices $v\in S$ where there is no other $v'\in S\cap v^{\Uparrow}$
strictly above $v$. It makes sense to hope that $\highest(L_{u})=\{l_{e}\}$
for any $u\in\hatP_{e}\setminus\{l_{e}\}$. This would be perfect
because, not only that the output size is small, the data structure even identifies
$l_{e}$ for us. Unfortunately, this cannot be true.
For example, let $x_{e}=\LCA(l_{e},\lca_{e})$, for any $u\in \upath{l_e}{x_e}\setminus \{l_e,x_e\}$ strictly between $l_e$ and $x_e$, $\highest(L_{u})$ might not contain $l_e$ because there might exist another mincut $u'^{\downarrow}\setminus l'^{\downarrow}$ such that 
$l_{e},l',u,u',x_{e}$ are ancestors of each other in this order and so $l_{e}\notin\highest(L_{u})$
because of $l'$.
Similarly, for any $u\in \upath{u_e}{x_e}\setminus \{u_e,x_e\}$ strictly between $u_e$ and $x_e$,  $\highest(L_{u})$ might not contain $l_e$ as well (see \Cref{fig:query-xe-must-find-le}).
So we could only hope to guarantee that $\highest(L_{u})=\{l_{e}\}$
when $u = x_{e}.$

\begin{figure}[h]
\centering
\includegraphics[]{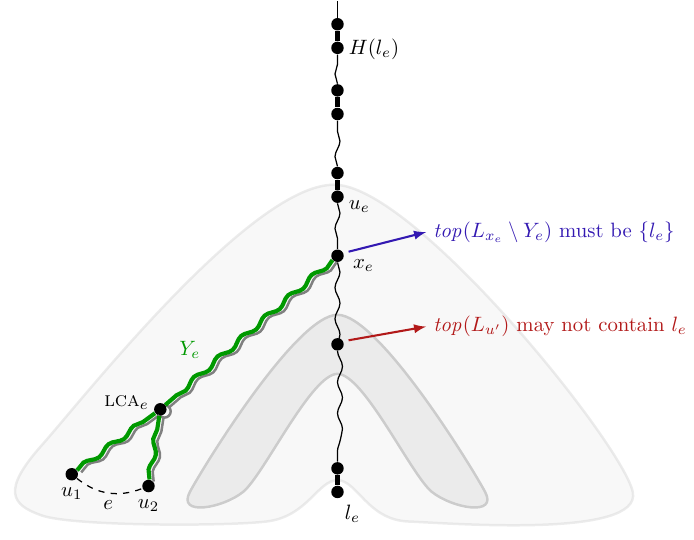}
\caption{Querying $x_e$ must find $l_e$. Querying at other vertices might not find $l_e$.}\label{fig:query-xe-must-find-le}
\end{figure}

Surprisingly, something very close to the above wishful claim is actually
true:
\begin{lem}
[Pinpoint the lower vertex] \label{lem:uniqueness-of-lower-vertex}
Suppose $e=(u_{1},u_{2})$ is a right edge whose minimal mincut is
$u_{e}^{\downarrow}\setminus l_{e}^{\downarrow}$. For vertex $x_e=\LCA(l_{e},\lca_{e})$, we have $\highest(L_{x_e}\setminus Y_e)=\{l_{e}\}$, where $Y_e := \upath{u_{1}}{u_{2}} \cup\upath{x_e}{\lca_e}$ is the \emph{avoiding} set of $e$.
\end{lem}
Before showing the proof, let us discuss the purpose of \Cref{lem:uniqueness-of-lower-vertex}. It helps us pinpoint the lower vertex $l_{e}$ because it says that, %
the lower vertex $l_{e}$ is exactly the unique top vertex of $L_{x_e}\setminus Y_e$,
the set of valid lower vertices below $x_e$ excluding $Y_e$.
Note that it is very natural to exclude $Y_e$ because,
for any right edge $e=(u_{1},u_{2})$, $l_{e}$ cannot be on $Y_e$,
otherwise $u_{e}^{\downarrow}\setminus l_{e}^{\downarrow}$ would
not contain $e$.

\begin{proof}[Proof of \Cref{lem:uniqueness-of-lower-vertex}]
Suppose for contradiction that $\highest(L_{x_e}\setminus Y_e)\neq\{l_{e}\}$.
First, observe that $\highest(L_{x_e}\setminus Y_e)$ is not empty since the set $L_{x_e}\setminus Y_e$ contains the vertex $l_e$. Now we arbitrarily select a vertex $l'\in \highest(L_{x_e}\setminus Y_e)$ different from $l_e$. By the definition of valid lower vertex below $x_e$, there exists $u'\in x_e^\uparrow$ that $u'^\downarrow\setminus l'^\downarrow$ forms a mincut. Since $l'\in L_{x_e}\setminus Y_e$ and $u'\in x_e^\uparrow$, the mincut $u'^\downarrow\setminus l'^\downarrow$ also contains edge $e$.
    
    Next we consider the relative position of these two mincuts $u_e^\downarrow\setminus l_e^\downarrow$ and $u'^\downarrow\setminus l'^\downarrow$. Since $l'$ is a top vertex of $L_{x_e}\setminus Y_e$, we have $l'\notin l_e^\downarrow$. So there are only two cases: $u'^\downarrow\setminus l'^\downarrow\subset u_e^\downarrow\setminus l_e^\downarrow$ or these two mincuts cross with each other.
    For the first case, $u'^\downarrow\setminus l'^\downarrow$ is a smaller mincut containing $e$, contradicts to minimality. For the second case, by the crossing property of mincuts (\Cref{lem:crossing-mincuts}), the intersection of these two mincuts is a smaller mincut containing $e$, again contradicts to the minimality of $u_e^\downarrow\setminus l_e^\downarrow$.
\end{proof}

Let $Y_{e,u} := \upath{u_{1}}{u_{2}} \cup\upath{u}{\lca_e}$. Note that $Y_e = Y_{e,{x_e}}$. Recall that our plan is to perform a post-order traversal. By \Cref{lem:uniqueness-of-lower-vertex}, whenever we arrives at $u = x_e$ and query $\highest(L_{u}\setminus Y_{e,u})$, the lower vertex $l_e$ would be returned for us and we are done for $e$. 
But we do not know $x_e$. Then, from which vertex $u$ should we query $\highest(L_{u}\setminus Y_{e,u})$ for finding $l_e$? Since we know that at least $x_e \in \lca_e^\uparrow$, we can query for finding $l_e$ when $u \in \lca_e^\uparrow$.
Now, the corollary below will be helpful because it says that for all $u \in \upath{x_e}{\lca_e}$, whenever we query for finding $l_e$, we will actually see nothing before the traversal actually reaches $x_e$.

\begin{corollary}\label{cor:uniqueness-of-lower-vertices}
For any vertex $u\in \upath{x_e}{\lca_e}\setminus \{x_e\}$, we have $\highest(L_{u}\setminus Y_{e,u})=\emptyset$.
\end{corollary}

\begin{proof}
Suppose for contradiction that there exists vertex $u\in \upath{x_e}{\lca_e}\setminus \{x_e\}$ such that $\highest(L_{u}\setminus Y_{e,u})\neq \emptyset$. Since $l_e\notin u^\downarrow$, the set $\highest(L_{u}\setminus Y_{e,u})$ contains other vertex $l'$ different from $l_e$. Since $\highest(L_{x_e}\setminus Y_e) = \{l_e\}$ by \Cref{lem:uniqueness-of-lower-vertex}, $H(l')\in x_e^\downarrow$. But this implies that $H(l')^\downarrow\setminus l'^\downarrow$ is a mincut containing $e$ with smaller size than $u_e^\downarrow\setminus l_e^\downarrow$, contradicts to minimality. %
\end{proof}

Equipped with this insight, now we are ready to move our attention on how to implement
our high-level approach efficiently.

\paragraph{Implementation.}

There are two main challenges in implementing the above high-level
approach. 
\begin{enumerate}
\item \textbf{(Efficiency of queries):} For any fixed vertex $u$, how can
we return $\highest(L_{u}\setminus Y_{e,u})$ quickly
given an edge $e$ as a query? Furthermore, as we perform
a post-order traversal, the vertex $u$ is not fixed. We need a dynamic
data structure where $u$ can be updated too.
\item \textbf{(The number of queries):} For any fixed edge $e$, if we query for $\highest(L_u \setminus Y_{e,u})$ on all $u$ or even just on all ancestors of $\lca_e$, then the total number of queries would be already super linear in $m$. We will exploit a structural lemma (\Cref{lem:general-case-property}) to reduce the number of queries.
\end{enumerate}

In order to describe our algorithm and address how do we cope with
both challenges, we first describe an algorithm that solves a simpler
case when $x_e=\lca_e$. That is, if we have a data structure that supports the queries to $\highest(L_{x_e}\setminus Y_e)$ then all lower vertices will be found by~\Cref{lem:uniqueness-of-lower-vertex}.
After we describe the algorithm that solves the simpler case, we generalize
the algorithm and solve both challenges in the general case.

\subsubsection*{Simple Case: {\normalfont{$\protect\lca_e$}} is always on the path
$\protect\hatP_{e}$.}

Let us assume here that, for every right edge $e$, $\protect\lca_e$ is on the path
$\protect\hatP_{e}$. Equivalently, $x_e={\lca_e}$. Even with this assumption, we will need to deal with the first challenge above.
We will remove this assumption in the next part.

We perform a post-order traversal on $T$. Suppose that $u$ is the current vertex. 
There are two main tasks: (1) we will show how to maintain
all the valid lower vertices below $u$, (2) we will show how to find the lower vertex $l_{e}$ for every right edge $e$
with $\lca_e = u$ via queries to the top-tree data structure.

To help solving the first task, we will exploit top-tree as follows.
The top-tree we are maintaining is always a subgraph (forest) of $T$, and each vertex is associated with a value that satisfying the following invariant.
Suppose the current vertex is $u$. For each vertex $v\in u^{\downarrow}$,
if $v$ is a valid lower vertex below $u$, the value of $v$
in the top-tree should be $depth_{v}$ (the depth of $v$ in $T$).
Otherwise, if $v$ is not a valid lower vertex below $u$, the value of $v$ is $\infty$.

Now we solve the first task.
The algorithm will maintain the invariant while running the
post-order traversal. Initially, the top-tree is the same as the spanning
tree with a super large value $\infty\gg n$ assigned to each vertex.
When the traversal reaches $u$, the algorithm will do the following
updates on the top-tree. First, the algorithm checks for each child $v$ of $u$:
if $H(v)\neq\NULL$, then assign value $depth_{v}$ to vertex $v$,
otherwise leave the value of $v$ unchanged (which is $\infty$).
Second, for any vertex $v\in u^{\downarrow}$ such that $H(v)$ is
a child of $u$, the algorithm assigns $\infty$ to $v$.
These vertices can be preprocessed once $H(v)$ is found.
Furthermore, we can safely assign $\infty$ to $v$ because that 
$v$ will no longer be a valid lower vertex below $u$ or below any vertex reached later in the post-order traversal.
Since
for each vertex $v$ the value of $v$ is changed at most twice in
the algorithm, by~\Cref{lem:data structures} the top-tree can be maintained in $O(n\log n)$ total
time.

Now we solve the second task.
We show how to use the top-tree to find $l_{e}$ for each edge $e=(u_{1},u_{2})$ where $\lca_e=u$ upon the reaching $u$ in the post-order traversal.
Recall that by~\Cref{lem:uniqueness-of-lower-vertex}, the lower vertex $l_{e}$ is the unique vertex in $\highest(L_{x_e}\setminus Y_e) = \highest(L_u\setminus \upath{u_1}{u_2})$.
It implies that $l_e$ is the unique vertex with smallest value among all vertices in $u^\downarrow\setminus\upath{u_1}{u_2}$ stored in the top-tree.
Therefore, we first apply $\cut(u,\mathrm{parent}(u))$ to separate the subtree rooted
at $u$. Then the lower vertex can be found by $l_{e}=\MinNonPath(u_{1},u_{2})$.
Finally, we apply $\link(u,\mathrm{parent}(u))$ to restore the tree. See \Cref{fig:le-by-minnonpath} for an illustration.

\begin{figure}[h]
\centering
\includegraphics[]{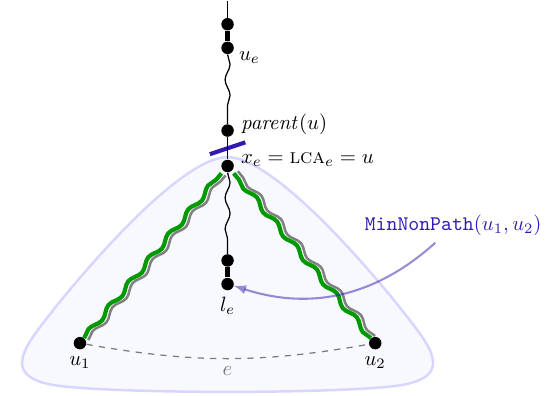}
\caption{Simple Case: $l_e$ can be found by $\MinNonPath(u_1, u_2)$ after separating the subtree rooted at $u=\lca_e$.}\label{fig:le-by-minnonpath}
\end{figure}

\subsubsection*{General Case: {\normalfont{$\protect\lca_e$}} may not be on the path
$\protect\hatP_{e}$.}

In the general case, we no longer have the assumption $\lca_e\in\hatP_{e}$.
By \Cref{lem:uniqueness-of-lower-vertex} and \Cref{cor:uniqueness-of-lower-vertices}, 
for each right edge $e$, 
it is natural to consider climbing up the tree from $\lca_e$ toward $x_e$.
The first time where the algorithm climbs up to a vertex $u$ such that $\highest(L_u\setminus Y_{e,u})\neq \emptyset$ implies that $u=x_e$.
However, the time
cost of performing such multiple queries per edge is unacceptable.

To deal with the challenge, we establish a key observation (\Cref{lem:general-case-property}) that leads to the following ``packaging'' idea. 
Initially every edge $e$ is individually packed and is assigned to the vertex $\lca_e$.
Upon reaching a vertex $v$ in the post-order traversal,
the data structure checks for each package whether or not a lower vertex can be assigned.
If a lower vertex $l$ is found, then all edges $e$ in the same package get the same lower vertex $l_e=l$. Otherwise, all packages will be combined into one large package and sent to the parent of $v$.
Our key observation states that, for all right edges $e$ so that
$v\in \hatP_{e}^\Leftrightarrow := \hatP_{e}^\downarrow \setminus 
(\hatP_e \cup l_e^\downarrow)$,
these right edges will be in the combined package and they all have the same minimal mincut. In particular, their lower vertices will be the same (see \Cref{fig:lemma-for-package}).%

\begin{figure}[h]
\centering
\includegraphics[]{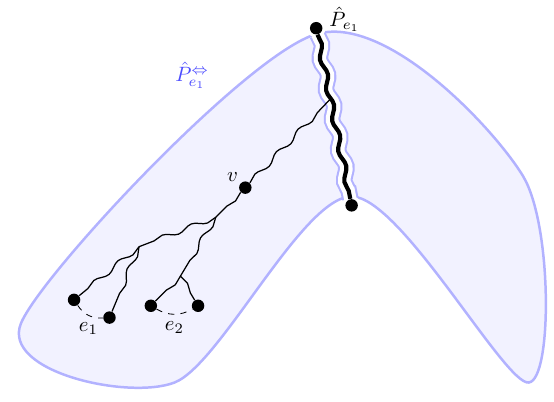}
\caption{An illustration to \Cref{lem:general-case-property}: if the algorithm arrives at vertex $v$ but has not found any lower vertex for $e_1$ and $e_2$ yet, then $\hat{P}_{e_1}=\hat{P}_{e_2}$.}\label{fig:lemma-for-package}
\end{figure}

\begin{lem}
\label{lem:general-case-property} Fix vertex $v$. For all right
edges $e_{1},e_{2}$ where their endpoints are in $v^{\downarrow}$
and $v\in\hatP_{e_{1}}^{\Leftrightarrow}\cap\hatP_{e_{2}}^{\Leftrightarrow}$
(i.e. $l_{e_{1}},l_{e_{2}}\not\in v^{\downarrow}$ and $u_{e_1},u_{e_2}\in v^\Uparrow$), we have that $\hatP_{e_{1}}=\hatP_{e_{2}}$.
\end{lem}

\begin{proof} %
First, we show that the minimal mincut of $e_1$ also contains $e_2$, and vice versa. Since $v\in \hatP_{e_1}^\Leftrightarrow$, the minimal mincut of $e_1$ contains the whole subtree $v^\downarrow$. Since both the endpoints of $e_2$ are in $v^\downarrow$, the minimal mincut of $e_1$ contains $e_2$. Symmetrically, the arguments also hold if we swap $e_1$ and $e_2$.
By~\Cref{lem:uniqueness-minimal}, the minimal mincut of $e_1$ and $e_2$ are the same and hence $\hatP_{e_1}=\hatP_{e_2}$.
\end{proof}

Now we implement the algorithm in the general case via top-tree.
When the traversal reaches vertex $u$, there are several packages at $u$. For the single edge packages that are directly created in preprocessing, the same procedure from simple case works: we cut off the subtree rooted at $u$ and query $\MinNonPath$ for each of these edges.
For any combined package that was delivered from a child $v$,
since there is no valid lower vertex found in $v^\downarrow$,
we claim that it suffices to check $\highest(L_u\setminus v^\downarrow)$.%

To see why, first observe that $l_e \notin L_v$ because $l_e$ is never in $Y_{e,v}$ and we also know that $l_e \notin \highest(L_v \setminus Y_{e,v})$ (since $e$ is forwarded from $v$). 
Now, suppose that $l_e \in \highest(L_u)$. We claim that $l_e$ actually is in $\highest(L_u \setminus v^\downarrow)$. This is because if $l_e \in v^\downarrow$, then $l_e \in L_v$. But we already concluded above that $l_e \notin L_v$. This completes the claim that if $l_e \in \highest(L_u)$, then $l_e \in \highest(L_u \setminus v^\downarrow)$. Now, by \Cref{lem:uniqueness-of-lower-vertex}, if $l_e \in \highest(L_u)$, then $\{l_e\}\subseteq \highest(L_u \setminus v^\downarrow)\subseteq \highest(L_u \setminus Y_{e,u}) = \{l_e\}$. Therefore, $\highest(L_u \setminus v^\downarrow) = \{l_e\}$. %

Therefore, the algorithm separates the subtree induced from $u^\downarrow\setminus v^\downarrow$ by cutting off the edges $(u, \mathrm{parent}(u))$ and $(v, u)$. Then $\highest(L_u\setminus v^\downarrow)$ is obtained by querying $\MinTreeUp(u)$.
If $\MinTreeUp(u)$ returns a vertex $l$ with a value not equal to $\infty$, then all edges within the package from $v$ has their lower vertex assigned to $l$ by \Cref{lem:general-case-property}.
After processing all packages at $u$, the algorithm combines all edges where their lower vertices are not found yet into one package, which can be efficiently implemented with a linked list. Back to a high-level explanation, observe that although there can be many edges in the package from $v$ forwarded to its parent $u$ that we want to query. We only need to query \emph{once} for each package. This is how we resolve the second implementation challenge about the number of queries.

To analyze the runtime, we notice that the number of top-tree queries made is linear to the total number of packages that has ever created, which is $O(n+m)$
top-tree queries. 
Therefore, the algorithm
for computing $l_{e}$ for every edge $e$ can be done in $O(m\log n)$
total time.
However, from \Cref{sec:computing-highest-partner} we know that preprocessing $H(v)$ values takes $O(m\log^2 n)$ time.
Hence, the total time computing lower vertices for each edge takes $O(m\log^2 n)$ time,
which proves \Cref{lem:computing-lower-vertex}.

%% file: 4-incomparable.tex
\section{Incomparable 2-respecting Minimal Mincuts of Edges}

\label{sec:label-incomparable}

In this section, we present the algorithm computing the minimal mincut of edge when it is a incomparable 2-respecting mincut of $T$. Similar to the comparable case, for every edge $e\in E$ we call $T$ the \emph{right} tree for $e$ if the minimal mincut of $e$ is a incomparable 2-respecting mincut of $T$. (In this case we also call $e$ a \emph{right edge} in $T$.)

\begin{lemma}
    \label{lem:labeling-edges-incomparable}
    There is an algorithm that, given a graph $G=(V,E)$ and a spanning tree $T\in \T$, in total time $O(m\log^2 n)$ computes, for every edge $e \in E$, an unordered vertex pair $g(e)=(v_e,w_e)$ or $\NULL$ where $v_e,w_e \in V$
    with the following guarantee: if $T$ is the right tree for $e$, then
    $v_e^\downarrow \cup w_e^\downarrow$ is the minimal mincut of $e$.
\end{lemma}

Let $e = (u_1,u_2)$ be a right edge in $T$ whose minimal mincut is $v^\downarrow \cup w^\downarrow$ where $v\neq w$. 
To prove \Cref{lem:labeling-edges-incomparable},
there are two main cases to consider: (1) the endpoints are in different subtrees: $u_1\in v^\downarrow$ and $u_2\in w^\downarrow$, or
(2) both endpoints are in the same subtree, e.g. both $u_1,u_2\in v^\downarrow$. 

For the first case, we solve it based on the main lemma below. We will devote to most of this section for proving it.

\begin{lemma}
    \label{lem:label-incomparable-path}
    Let ${\cal P}$ be a path decomposition of $T$. There is an algorithm that preprocesses $G=(V,E)$
    and spanning tree $T$ in $O(m\log n)$ time so that, given any path $P\in \cP$, in time
    $O(d(P^{\downarrow})\log n)$ the algorithm computes, for every edge $e=(u_{1},u_{2})$
    such that $e\in E(P^\downarrow)$,
    an unordered vertex pair $g(e, P) =(v,w)$ or $\NULL{}$ with the following guarantee: 
    
    Suppose $v^{\downarrow}\cup w^{\downarrow}$ is the incomparable minimal mincut for $e$ such that $u_1\in v^\downarrow$ and $u_2\in w^\downarrow$. Let $P_v, P_w\in \mathcal{P}$ be the paths that contains $v$ and $w$ respectively.
    Then, either $g(e, P_v)=(v, w)$ or $g(e, P_w)=(v, w)$.
\end{lemma}

For the second case, observe that $\LCA(u_1,u_2)\in v^\downarrow$ since $u_1, u_2\in v^\downarrow$. Therefore, the minimal mincut of the right edge $e$ is also the \emph{minimal incomparable 2-respecting mincut} of vertex $\LCA(u_1,u_2)$, which can be computed efficiently using \Cref{lem:labeling-vertex-incomparable}. The proof of \Cref{lem:labeling-vertex-incomparable} is deferred to \Cref{sec:labeling-vertices}, since it also serves as a building block of \Cref{lem:labeling-vertices}.

\begin{lemma}
    \label{lem:labeling-vertex-incomparable}
    There is an algorithm that, given a spanning tree $T$ of $G=(V,E)$,
    in total time
    $O(m\log^2 n)$ computes, for every vertex $u\in V$ the minimal incomparable 2-respecting mincut candidate $f(u) = (v_u,w_u)$ or $\NULL$ with the following guarantee:
    
    If there exists an incomparable 2-respecting cut that separating $u$ from root $r$, then $f(u)\neq \NULL{}$ and $v_u^\downarrow\cup w_u^\downarrow$ is such a mincut with smallest size.
\end{lemma}

Now we prove \Cref{lem:labeling-edges-incomparable} using \Cref{lem:label-incomparable-path,lem:labeling-vertex-incomparable}:

\begin{proof}[Proof of \Cref{lem:labeling-edges-incomparable}]
    Suppose $v^{\downarrow}\cup w^{\downarrow}$ is the minimal mincut for a right edge $e = (u_1,u_2)$. 
    Depending on whether the endpoints of $e$ are in the same subtree (rooted in either $v$ or $w$), or in the different subtrees, 
    we consider the following  two cases:
    \begin{description}[itemsep=0pt]
        \item[Case 1: (different subtrees)] WLOG, assume $u_1\in v^\downarrow$ and $u_2\in w^\downarrow$. Given a balanced path decomposition $\cP$, for each path $P\in \cP$, the algorithm in \Cref{lem:label-incomparable-path} computes a candidate vertex pair $g(e,P)=(v',w')$ or $\NULL{}$ for every edge $e = (u_1,u_2)$ such that $e\in E(P^\downarrow)$.
        Plugging in the path decomposition framework \Cref{lem:reduc to path}, we can compute $g(e,P)$ for all edge $e = (u_1,u_2)\in E$ and all $P\in \cP$ such that $e\in E(P^\downarrow)$ in $O(m\log^2 n)$ time.%
        
        Next, for each edge $e$, 
        the algorithm compares the size of all candidate mincuts $g(e,P)$
        for each path $P$ such that $e\in E(P^\downarrow)$
        and select the smallest one to be $(v_e,w_e)$.
        This step takes $O(\log n)$ time for each edge $e$ since there are at most $O(\log n)$ paths $P$ such that $e\in E(P^\downarrow)$ by the balanced property of $\cP$.
        Since \Cref{lem:label-incomparable-path} guarantees that either $g(e,P_v)$ or $g(e,P_w)$ equals to $(v,w)$, our algorithm find the minimal mincut of  $e$.
        
        The total time for this case is $O(m\log^2 n) + O(m\log n) = O(m\log^2 n)$.

        \item[Case 2: (same subtree)] WLOG, assume $u_1, u_2\in v^\downarrow$. This implies $\LCA(u_1,u_2)\in v^\downarrow$. Therefore, the minimal mincut of $e$ will also be the minimal incomparable 2-respecting mincut of $\LCA(u_1,u_2)$, which has been computed via the algorithm from \Cref{lem:labeling-vertex-incomparable} in $O(m\log^2 n)$ time.
        
    \end{description}
    Since it takes $O(m\log^2 n)$ time for both cases, the whole algorithm runs in $O(m\log^2 n)$ total time.
\end{proof}

The rest of the section is devoted for proving \Cref{lem:label-incomparable-path}.
In~\Cref{sec:min-v-precut-and-p-outer-min-v-precut} we introduce several essential concepts that allow us to describe and prove the algorithm in a precise way.
In~\Cref{sec:incomparable-alg-desc} we describe the high-level idea to the algorithm that finds all incomparable minimal 2-respecting mincut candidate for each edge, and in \Cref{sec:incomparable implementation} we complete the implementation details.
Finally in \Cref{sec:incomparable correctness} and \Cref{sec:incomparable time analysis} we prove the correctness of the algorithm and analyze the runtime, concluding the proof of~\Cref{lem:label-incomparable-path}.

\subsection{Minimum $v$-Precuts and $P$-Outer Minimum $v$-Precuts}
\label{sec:min-v-precut-and-p-outer-min-v-precut}

Fix a spanning tree $T\in {\cal T}$. 
Our algorithm for \Cref{lem:label-incomparable-path} requires implementation on efficiently identifying incomparable 2-respecting mincuts $v^\downarrow\cup w^\downarrow$ with one vertex $v$ on the path $P\in\cP$.
In this section we first recall the concepts of \emph{minprecut values} from~\cite{karger2000minimum} and then we 
introduce \emph{outer minprecut values} that help us to describe the algorithm with clarity.

Observe that for any pair of incomparable vertices $v\perp w$, the weight of the cut induced by $v^\downarrow\cup w^\downarrow$ can be expressed as
\[
    \C(v^\downarrow\cup w^\downarrow) = \C(v^\downarrow) + \C(w^\downarrow) - 2\C(v^\downarrow,w^\downarrow) ~.
\]

Suppose $v^\downarrow$ is one \emph{part} (the other part is $w^\downarrow$) %
of the minimal incomparable 2-respecting mincut of $e$.
Then by factoring out the terms only related to $v$, it suffices to compute the following \emph{incomparable precut values} for all $w\perp v$.

\begin{definition}[Incomparable precut value]
    The \emph{incomparable precut value} of $v$ at $w$, is defined by 
    \[
        \CI_v(w) := \C(w^\downarrow) - 2\C(v^\downarrow,w^\downarrow) ~.
    \]
    Note that the value is only defined in the incomparable scenario $w\perp v$.%
\end{definition}

The incomparable precut values are defined analogously to the comparable precut values. Actually, they share the same functionality in the sense that given a vertex $v$, an algorithm can be designed to find a \emph{partner} $w$ such that $v$ and $w$ together idetifies a incomparable (resp. comparable) 2-respecting mincut $v^\downarrow\cup w^\downarrow$ (resp. $w^\downarrow\setminus v^\downarrow$). Such partner should be a minimizer of the \emph{minpercut value} defined below.

\begin{definition}
    The \emph{incomparable minprecut value} of $v$, is defined by
    \[
        \CI_v := \min\{ \CI_v(w)\mid w\perp v \} ~.
    \]
\end{definition}

A vertex $w$ is called an \emph{incomparable mincut partner} of $v$ if $v^\downarrow\cup w^\downarrow$ is a mincut. For brevity, we will omit the word ``comparable/incomparable'' and simply call $\CI_v(w)$ as precut value of $v$ (at $w$) and call $w$ a \emph{mincut partner} of $v$ if the context is clear. Besides, we call $w$ a \emph{minprecut partner} of $v$ if $\CI_v(w) = \CI_v$. Note that a mincut partner must be a minprecut partner, but it may not be correct conversely.

Given a path $P\in \cP$, our algorithm will maintain the minprecut value of the current vertex $v$ when climbing up the path $P$. It turns out that if we further exclude all the candidates of partner within $P^\downarrow$, the algorithm can maintain minprecut values and partners of $v$ in a more efficient way, which leads to the following definitions.

\begin{definition}
    We call $w$ an outer vertex of path $P$ if $w\not\in P^\downarrow$.
\end{definition}

The \emph{outer minprecut value} is defined similar to the minprecut value except that only the outer vertices are taken in consideration.

\begin{definition}
    \label{def:outer-minprecut}
    The \emph{$P$-outer minprecut value} of $v$ is defined as the minprecut value of $v$ such that the minprecut partner $w$ is an outer vertex of $v$. Specifically,
    \[
        \CI_{v,P} = \min\{\CI_v(w)\mid w\notin P^\downarrow, w\perp v\}
    \]
\end{definition}

We call $w$ a \emph{$P$-outer minprecut partner} of $v$ if $w$ is an outer vertex of $P$ and $\CI_{v,P}(w) = \CI_{v,P}$.
Now we are ready to describe the algorithm.

\subsection{Algorithm Description}
\label{sec:incomparable-alg-desc}
Now, we are ready to describe the algorithm for \Cref{lem:label-incomparable-path}.

The high-level description of the algorithm is as follow: 
Given a path $P = (v_1,v_2\cdots,v_k)\in \cP$ with $v_1$ being the deepest vertex, our algorithm visits $v_i$ in the order of $i=1, 2, \ldots, k$.
At iteration $i$, 
the algorithm visits $v_i$ and maintains the invariant 
such that the values $\CI_{v_i}(w)$ for all $P$-outer vertex $w$ can be accessed via $val[w]$, and the $P$-outer minprecut value $\CI_{v_i, P}$ is stored in $val^*$.

Whenever the algorithm reaches the vertex $v_i$, it first calls a subroutine called $\LocalUpdate(v_i)$ that will recover the invariant on $val[w]$. %
Once the invariant holds, if $v_i$ has a $P$-outer mincut partner $w$,
then any edge $e = (u_1,u_2)$ where $u_1\in v_i^\downarrow$ and $u_2\in w^\downarrow$ which has not been assigned a minimal mincut candidate yet should obtain an incomparable 2-respecting mincut $(v_i, w')$ for some specific choice of $w'\in w^\downarrow$ as the minimal mincut candidate $g(e,P)$.

To implement this high-level plan,  our algorithm will maintain a set $\hat{E}\subseteq E(v_i^\downarrow)$
that contains all edge $e$ whose $g(e, P)$ values is not assigned yet,
and a \emph{witness} set $W$ of vertices that will be helpful for extracting the correct minimal mincut candidates. We summarize the invariant for $W$ as follows, and we defer the proof to the end of \Cref{sec:incomparable correctness}.

\begin{lemma}[Invariant for $W$ and $\hat{E}$]\label{lem:invariant-of-w}
The witness set $W$ and the edge set $\hat{E}$ satisfies the following invariant whenever the algorithm returns from \textup{$\LocalUpdate(v_i)$} when visiting $v_i\in P$:
\begin{enumerate}[itemsep=0pt]
    \item[(1)] $W$ is always a subset of $P$-outer minprecut partners of the current visiting vertex $v_i$.
    \item[(2)] Any edge $e=(u_1, u_2) \in \hat{E}$ satisfies that $u_1\in v_i^\downarrow$ and $u_2\notin P^\downarrow$.
    \item[(3)] \textbf{(Correctness Guarantee):} for any edge $e=(u_1,u_2)\in \hat{E}$ such that $u_1\in v_i^\downarrow$, if there exists a $P$-outer mincut partner of $v_i$ that is an ancestor of $u_2$, then there exists a witness $w\in W$ that is an ancestor of $u_2$.
\end{enumerate}
\end{lemma}

\paragraph{Use $W$ to find incomparable minimal 2-respecting mincuts for edges.}

Following this property, we are able to present the algorithm to find incomparable minimal 2-respecting mincuts for edges using $W$ and $\hat{E}$.
The algorithm visits $v_1, v_2, \ldots, v_k$ along the path.
Whenever the algorithm visits $v_i$, all edges $e=(u_1, u_2)$ such that $u_1\in v_i^\downarrow\setminus v_{i-1}^\downarrow$ and $u_2\notin P^\downarrow$ shall be added to $\hat{E}$ (let $v_0^\downarrow = \emptyset$ for convenience).

In order to obtain correct results,
after $v_i$ is visited,
the algorithm should identify a mincut candidate $g(e, P):=(v_{e, P}, w_{e, P})$ for all right edges $e\in \hat{E}$ with $v_{e,P}=v_i$.
One way to achieve this is to scan through all edges in $\hat{E}$ and 
use $\MinPDown$ queries on a dynamic tree to find $w_{e, P}$ (whenever $v_{e, P}=v_i$ we will find $w_{e, P}$ correctly).
However, the amount of edges whose $v_{e, P}=v_i$ could be a tiny fraction in $\hat{E}$.

Using the correctness guarantee from \Cref{lem:invariant-of-w},
we will be able to identify all edges whose $v_{e, P}=v_i$ without spending any time on any other edges: all we need to do is to extract the subset of $\hat{E}$ whose outer endpoints are descendants of any $w\in W$.%
Interestingly, the task of finding all descendant outer endpoints can be reduced to a dynamic range query problem, which can be solved efficiently using a standard binary search tree. We discuss the implementation details in \Cref{sec:incomparable implementation} and also in \texttt{AssignMinCut$(v_i)$} (see~\Cref{alg:mincut subtree update}).

Now it comes down to efficiently maintain the witness set $W$.

\paragraph{Maintaining the set $W$.}
The algorithm maintains the set $W$ as following:
\begin{itemize}[itemsep=0pt]
    \item
    In the beginning of iteration $i$, if the minprecut value of $v_i$ is the same of the minprecut value of $v_{i-1}$, we keep the set $W$. Otherwise reset $W$ to be empty.
    \item
    For each newly added edge $e = (u_1, u_2)$ where $u_1\in v_i^\downarrow\setminus v_{i-1}^\downarrow, u_2\notin P^\downarrow$, we find the highest ancestor of $u_2$ which is a $P$-outer minprecut partner of $v_i$, and add it to $W$.
    \item
    For the current vertex $v_i$, if it has some $P$-outer mincut partners, reset $W$ to be empty in the end of this iteration.
\end{itemize}

We summarize and implement  the procedure that maintains $W$, together with maintaining the precut values as we move from the vertex $v_{i-1}$ to $v_i$ into the algorithm  \texttt{LocalUpdate$(v_i)$}.
The high level implementation of the entire algorithm is summarized in \Cref{alg:computing incomparable}.

\begin{algorithm}[H]
\caption{Computing incomparable 2-respecting minimal mincuts of edges}
\label{alg:computing incomparable}
\begin{algorithmic}[1]
\State Initialize the witness set $W\gets \emptyset$, $\hat{E}\gets\emptyset$, and $val^*=\infty$.

\For{\textbf{each} $i=1, 2, \ldots, k$}\Comment{The $i$-th iteration handles $v_i$.}

\State Call \texttt{LocalUpdate$(v_i)$}. \Comment{\Cref{alg:local update}: update $W$, $\hat{E}$, and $val^*$.}

\State Call \texttt{AssignMinCut$(v_i)$}. \Comment{\Cref{alg:mincut subtree update}.}

\If{$\CI(v_i^{\downarrow}) + val^* = \lambda$}
\State Reset $W \gets \emptyset$.
\EndIf
\EndFor
\end{algorithmic}
\end{algorithm}

\subsection{Implementation}
\label{sec:incomparable implementation}

There are two steps that we need to provide a detailed implementation: maintaining the set $W$ (implementing $\texttt{LocalUpdate}(v_i)$), and using $W$ to compute the minimal 2-repsecting mincut candidates of all edges via range queries (implementing $\texttt{AssignMinCut}(v_i)$).

\paragraph{Implementing \texttt{LocalUpdate} and Preprocessing.}

We use a dynamic tree (\Cref{lem:data structures}) to maintain the precut values of the current vertex $v_i$ while climbing up the path $P$. We maintain the invariant that once we process the vertex $v_i$ the values  $\CI_{v_i}(w)$ for all $w\perp v_i, w\notin P^\downarrow$ can be accessed via $val[w]$ .
Recall that $\CI_{v_i}(w) = \C(w^\downarrow) - 2\C(v_i^\downarrow,w^\downarrow)$. In the preprocessing step before the path $P$ was given, we set $val[w] = \C(w^\downarrow)$ for each vertex $w$ and create the dynamic tree on $T$, which can be done using the same preprocessing step of comparable case in \cref{note:preprocess}. After that, we initialize the minprecut value $val^*$ to be the minimum of $val[w]$ over all the $P$-outer vertices, which can be done using the following dynamic tree operations. We first apply $\cut(v_n,\mathrm{parent}(v_n))$ to separate the subtree $P^\downarrow$. Then we apply $\AddP(\mathrm{parent}(v_n), \infty)$ to set the value of all the ancestor of $v_n$ to be $\infty$, since they are not incomparable vertices of any vertex $v_i\in P$. Now we use $\MinTreeUp(\mathrm{parent}(v_n))$ to find $w$ which is argmin of $val[w]$ over all the $P$-outer vertices, and set $val^* = val[w]$.
Finally, we apply $\link(v_n,\mathrm{parent}(v_n))$ to restore the tree.

Given the path $P$, the algorithm scans through the vertices $v_1,v_2,\cdots v_k$ on the path one by one. Suppose the algorithm reaches $v_i$ now.
With the invariant after processing $v_{i-1}$ (the invariant for $v_0$ is $val[w] = \C(w^\downarrow)$),
it suffices to substract $2\C(v_i^\downarrow\setminus v_{i-1}^\downarrow, w)$ to $val[w]$ for each $w\perp v_i, w\notin P^\downarrow$ by invoking $\AddP(u, -2\C(v, u))$ for each edge $(v, u)$ where $v\in v_i^\downarrow\setminus v_{i-1}^\downarrow, u\notin P$. (These edges can be found in $O(d(v_i^{\downarrow}\setminus v_{i-1}^\downarrow))$ time using a DFS from $v_i$ without searching the subtree rooted at $v_{i-1}$.) 
Therefore, in $O(d(v_i^\downarrow\setminus v_{i-1}^\downarrow)\log n)$ time, $val[w]$ are updated to $\CI_{v_i}(w)$ for all $P$-outer vertex $w\perp v_i$.

While the algorithm updates the precut value as stated above, it also updates the minprecut value $val^*$, set $W$ and set $\hat{E}$. The algorithm scans through the edges $(v, u)$ where $v\in v_i^\downarrow\setminus v_{i-1}^\downarrow, u\notin P$ and invokes $\MinPUp(u)$ to find the vertex $w$ with the minimum precut value, and break tie by finding the highest one. Then we update the minprecut value $val^*$. If $val^*$ changes, we reset $W$ to be empty. Finally, if $val[w] = val^*$, we insert $w$ to the set $W$. These steps for maintaining the set $W$ are equivalent to the high-level description in \Cref{sec:incomparable-alg-desc}.

Notice that in \Cref{alg:local update} we explicitly use a global binary search tree data structure (BST) representing the set $\hat{E}$.
Thus, there will be some steps (line~\ref{code:insert BST}) involving BST that will be used for finding the mincut candidates, which is described in the next paragraph.

\begin{algorithm}[H]
    \caption{\texttt{LocalUpdate$(v_i)$}}
    \label{alg:local update}
    \begin{algorithmic}[1]
        \State Call \texttt{AddPath}$(v_i,\infty)$.
        \label{code:set-ancestor-infinity}
        \For{\textbf{each} edge $e=(v,u)$ such that $v\in v_i^\downarrow\setminus v_{i-1}^\downarrow, u\notin P^\downarrow$}
            \State Call $\InsertBST(\PreOrder[u], e)$. \Comment{Insert $e$ to $\hat{E}$.}
            \label{code:insert BST}
            \State Call \texttt{AddPath}$(u,-2\C(v,u))$.
            \label{code:add-path}
            \State Call $w\gets \MinPUp(u)$.
            \If{$val[w]<val^*$}
            \State $val^* \gets val[w]$.
            \State Reset $W \gets \emptyset$.
            \EndIf
            \If{$val[w] = val^*$}
            \State Insert $w$ into $W$.
            \label{line:insert w}
            \EndIf
            \label{code:outer-minprecut-value-update}
        \EndFor
        \label{code:local-update-W-reset}
        \label{code:W-update}
    \end{algorithmic}
\end{algorithm}

\paragraph{Computing minimal mincut candidates of edges from $W$.}
First, we clarify the global data structures in the algorithm. While climbing up the path $P$, we use a binary search tree (BST) to store all the edges in $\hat{E}$, with the key equals to the pre-order indices of the outer endpoints. Note that there will be only one global BST in our algorithm. The BST supports the following primitives.

\begin{itemize}
    \item $\InsertBST(x, e)$, store edge $e = (u_1,u_2)$ ($u_1\in P^\downarrow, u_2\notin P^\downarrow$) into the BST with key $x$.
    In our algorithm, we will always set $x$ to be $\texttt{pre\_order}[u_2]$, i.e., the pre-order index of vertex $u_2$.
    \item $\ExtractBST(w^{\downarrow})$, extract all the edges with endpoints in $w^\downarrow$, i.e. the edges with index in the interval corresponds to $w^\downarrow$.
\end{itemize}

We use the BST to store the edges that have not found their minimal mincut candidates yet.
When the algorithm reaches $v_i$ on the path, it inserts all edges $e=(v, u)$ with $v$ in $v_i^\downarrow\setminus v_{i-1}^\downarrow$ and $u\notin P^\downarrow$ to the BST, with the key $\PreOrder[u]$. This step is implemented in line~\ref{code:insert BST} of \Cref{alg:local update}.
To compute the minimal mincut candidates of edges using $W$, we first check whether there exists a mincut partner of $v_i$ (line~\ref{code:check mincut} in \Cref{alg:mincut subtree update}). If there exists, we perform the range query using BST to extract all the edges with endpoint $u\in w^\downarrow$ for each $w\in W$, which is implemented in \Cref{alg:mincut subtree update}.
Finally, for each extracted edge with endpoint $u\in w^\downarrow$, we find $w'$ to be the lowest ancestor  such that $w'$ is the mincut partner of $v_i$, and set $g(e,P)=(v_i,w')$.

\begin{algorithm}[H]
    \caption{\texttt{AssignMinCut$(v_i)$}}
    \label{alg:mincut subtree update}
    \begin{algorithmic}[1]
        \If{$\C(v_i^{\downarrow}) + val^* = \lambda$}
        \label{code:check mincut}
            \For{\textbf{each} $w\in W$}
                \For{\textbf{each} $e = (v,u)\gets \ExtractBST(w'^{\downarrow})$}
                    \State $w'\gets \MinPDown(u)$.
                    \State Set $g(e,P) = (v_i,w')$.
                \EndFor
            \EndFor
        \EndIf
    \end{algorithmic}
\end{algorithm}

\subsection{Correctness}

\label{sec:incomparable correctness}

First, we show a basic property of path decomposition and outer vertex that will be useful for analyzing the algorithm.

\begin{lemma}
\label{lem:outer property}
    For any spanning tree $T$, any path decompostion $\cP$, and any vertices $v$ and $w$ so that $v\perp w$, let $P_v\neq P_w\in\cP$ be the paths that contains $v$ and $w$ respectively. Then either $v$ is a $P_w$-outer vertex or $w$ is a $P_v$-outer vertex.
\end{lemma}

\begin{proof}
Observe that either $P_v\cap  P_w^\downarrow = \emptyset$ or $P_w\cap P_v^\downarrow = \emptyset$, which implies the lemma.
\end{proof}

For any right edge $e = (u_1, u_2)$, suppose the minimal mincut of $e$ is $v_e^\downarrow\cup w_e^\downarrow$ where $u_1\in v_e^\downarrow$ and $u_2\in w_e^\downarrow$.
By \Cref{lem:outer property}, without loss of generality, suppose $w$ is a $P_v$-outer vertex where $P_v$ is the path containing $v$.
The main property for the correctness of the algorithm can be summarized in \Cref{lem:incomparable-key-invariant}, which can be derived from maintaining the invariant of $W$ (\Cref{lem:invariant-of-w}).

\begin{corollary}
\label{lem:incomparable-key-invariant}
    For a path $P\in \cP$ and an edge $e = (u_1,u_2)$ with $u_1\in P^\downarrow$ and $u_2\notin P^\downarrow$, let $v_i$ be the deepest vertex in $P$ that has an outer mincut partner $w$ where 
    $u_1 \in v_i^\downarrow$ and $u_2 \in w^\downarrow$. 
    Then, there exists an ancestor of $w$ in $W$ when returned from \textup{$\LocalUpdate(v_i)$} at iteration $i$ in \Cref{alg:computing incomparable}.
\end{corollary}

\begin{proof}
It suffices to show that $e\in \hat{E}$ at iteration $i$ after returning from $\LocalUpdate(v_i)$. Then the lemma is implied by the correctness guarantee of \Cref{lem:invariant-of-w}.

Since $v_i$ is the deepest vertex in $P$ that has an outer mincut partner $w$ where $u_2 \in w^\downarrow$, the edge $e$ has not obtained its mincut candidate $g(e,P)$ yet, hence $e\in \hat{E}$.
\end{proof}

Using \Cref{lem:incomparable-key-invariant}, we are able to prove that all minimal incomparable 2-respecting mincut candidate can be correctly found:

\begin{proof}[Proof of the correctness of \Cref{lem:label-incomparable-path}]
    For any right edge $e = (u_1,u_2)$, suppose the minimal mincut of $e$ is $v^\downarrow\cup w^\downarrow$ where $u_1\in v^\downarrow$ and $u_2\in w^\downarrow$. By \Cref{lem:outer property}, WLOG, suppose $w$ is a $P_v$-outer vertex where $P_v$ is the path containing $v$. First, we show that there doesn't exist any mincut $v'^\downarrow\cup w'^\downarrow$ containing $e$ such that $v'\in v^\Downarrow$. Otherwise, the intersection of $v'^\downarrow\cup w'^\downarrow$ and $v^\downarrow\cup w^\downarrow$ will be a smaller mincut containing $e$, contradicts to minimality.
    
    Let $P = P_v$ for brevity. %
    By the argument above, $v$ is the deepest vertex in $P$ that has an outer mincut partner $w$ where $u_1\in v_i^\downarrow$ and $u_2\in w^\downarrow$. Then by \Cref{lem:incomparable-key-invariant}, there exists $w'$ to be an ancestor of $w$ such that $w'\in W$ at the iteration of $v$ in \Cref{alg:computing incomparable}. Therefore, the algorithm will extract the edge $e$ since $u_1\in v^\downarrow$ and $u_2\in w'^\downarrow$.
    Then the algorithm finds $w$ as it is the lowest ancestor of $u_2$ and also a $P$-outer mincut partner of $v$. Finally, the algorithm sets $g(e,P) = (v,w)$ as desired.
\end{proof}

Finally, we prove that the invariant for $W$ and $\hat{E}$ holds.

\begin{proof}[Proof of~\Cref{lem:invariant-of-w}]
The first and second invariants are directly from the algorithm.

Now we prove the third invariant correctness guarantee. Consider when the algorithm returns from \textup{$\LocalUpdate(v_i)$}. For any edge $e = (u_1,u_2)\in \hat{E}$ where $u_2\notin P^\downarrow$ such that there exists a $P$-outer mincut partner of $v_i$ that is an ancestor of $u_2$, let $w$ be the lowest $P$-outer mincut partner of $v_i$ that is an ancestor of $u_2$, then we prove that $w\in W$ . 

Observe that $val^*$ is non-increasing over the time. Let $j$ be the largest index such that $\CI_{v_j,P}(w)\neq \CI_{v_{j-1},P}(w)$. Since $w$ is a mincut partner of $v_i$, we deduce that $w$ is inserted to $W$ in $\LocalUpdate(v_j)$ and $val^*$ doesn't change during the iterations $(j,i]$, which implies that $W$ is incremental in the following iterations. Therefore, $w$ remains in $W$ when the algorithm returns from \textup{$\LocalUpdate(v_i)$}.
\end{proof}

\subsection{Runtime Analysis}
\label{sec:incomparable time analysis}
\begin{proof}[Proof of the running time of \Cref{lem:label-incomparable-path}]

First, for the preprocessing step stated in the beginning of \Cref{sec:incomparable implementation}, it takes $O(m\log n)$ time to initialize $val[w] = \C(w^\downarrow)$ for every vertex $w$ and create the dynamic tree on $T$. After given a path $P$, it takes constant dynamic tree operations to initialize $val^*$.

In iteration $i$ of \Cref{alg:computing incomparable}, we show that the cost of the subroutines $\texttt{LocalUpdate}(v_i)$  is in $O(d(v_i^\downarrow\setminus v_{i-1}^\downarrow)\log n)$ time. And we show that the total time processing $\texttt{AssignMinCut}(v_i)$ is $O(d(P^\downarrow)\log n)$ for the whole path $P$. Therefore, we deduce that the total time processing a path $P$ is within $O(d(P^\downarrow)\log n)$ time.

In \Cref{alg:local update} $\texttt{LocalUpdate}(v_i)$, for each edge $e$ with one endpoint in $v_i^\downarrow\setminus v_{i-1}^\downarrow$, the algorithm invoke constant time dynamic tree operations and BST operations. Therefore, the total time cost is $O(d(v_i^\downarrow\setminus v_{i-1}^\downarrow)\log n)$.

In \Cref{alg:mincut subtree update} $\texttt{AssignMinCut}(v_i)$, it only goes into the if clause when there exists some mincut partner of $v_i$. Then the time cost will be $O(p_i+q_i\log n)$ where $p_i$ is the size of set $W$ and $q_i$ is the number of edges extracted from the BST in iteration $i$.
Since $W$ is reset to be empty each time when the if-clause is executed, and by line~\ref{line:insert w} in \Cref{alg:local update} each edge in $E(P^\downarrow)$ causes at most one insertion to $W$, hence $\sum_i p_i$ is bounded by $O(d(P^\downarrow))$.
Furthermore, each edge got inserted and deleted at most once in $\hat{E}$ so $\sum_i q_i$ is bounded by $O(d(P^\downarrow))$.

Therefore, the algorithm preprocesses $G=(V,E)$ and spanning tree $T$ in $O(m\log n)$ time, and processes each path $P$ in $O(d(P^\downarrow)\log n)$ time.
\end{proof}

%% file: put-together.tex
\section{Putting Everything Together: Proof of \Cref{lem:labeling-edges}}
\label{sec:put-together}

Finally, we show that given a spanning tree $T$, we can combine the results in \Cref{sec:label-comparable,sec:label-incomparable} to get the minimal 2-respecting mincut of every edge $e$ in $O(m\log^2 n)$ total time. As the discussion below \Cref{lem:labeling-edges}, the mincut can be classified into three types.

\begin{description}[itemsep=0pt]
\item[Type 1.] By \Cref{lem:1-respecting}, the algorithm computes the minimal 1-respecting mincut candidate for every edge $e$.
\item[Type 2-Comparable.] By \Cref{lem:labeling-edges-comparable}, the algorithm computes the minimal comparable 2-respecting mincut candidates for every edge $e$.
\item[Type 2-Incomparable.] By \Cref{lem:labeling-edges-incomparable}, the algorithm computes the minimal incomparable 2-respecting mincut candidates for every edge $e$.
\end{description}

For each edge $e$, we get three minimal mincut candidates as above. If all of the three candidates are $\NULL$, then the minimal 2-respecting mincut of $e$ respects to $T$ is $\NULL$. Otherwise, the minimal 2-respecting mincut of $e$ respects to $T$ is the mincut with the smallest size among the minimal mincut candidates.

Since the algorithm for each case runs in $O(m\log^2 n)$ time and the comparing time is constant for each edge, the whole algorithm runs in $O(m\log^2 n)$ total time, concluding~\Cref{lem:labeling-edges}.

\section*{Acknowledgment}
We thank David Karger and Debmalya Panigrahi for their clarification related to \cite{karger2009near}.

%% file: A-app-construction.tex
\section{Constructing Cactus from Minimal Mincuts}
\label{app:reduction-algorithm}

The goal of this section is to prove \Cref{thm:reduction}:

\cactusconstructionlemma*

The difference between \Cref{thm:reduction} and the statement given in \cite{karger2009near} is that we are given all cut labels representing minimal mincuts of \emph{all} edges, while in \cite{karger2009near} they are given only the labels of minimal mincuts of some edges.\footnote{In particular, \cite{karger2009near} reduces the problem into several \emph{one-layer cactus construction problems}, and piecing back all constructed partial (and contiguous) cactus  after solving these separated subproblems.
The main reason for introducing this seemly-extra reduction step
is because in their algorithm some edges do not have correct minimal mincut labels --- these erroneous labelings will not cause a problem (with high probability) in the reduced problems.
This reduction makes the problem simpler to solve.
However, the reduction  complicates the correctness proof and some details were omitted. 

Our simplified algorithm described in this section does not depend on the reduction mentioned above. We emphasize that this algorithm is in fact equivalent to performing \cite{karger2009near} in a bottom-up fashion. We give the algorithm and the correctness proof via a \emph{hierarchy representation} introduced implicitly by Gabow~\cite{gabow2016minset} in \Cref{sec:one level reduc}.}
Because of this, we are able to give a simpler algorithm than the one provided in \cite{karger2009near}.

This section is organized as follows. 
In \Cref{sec:gabow-chain} and \Cref{sec:cactus-hierarchy} we define the
\emph{hierarchy representation} $H$ for all mincuts.
With the hierarchy representation $H$, a cactus can then be constructed in linear time (\Cref{lem:hierarchy-to-cactus}).
In \Cref{sec:cactus-algorithm} we describe the algorithm that computes the hierarchy representation $H$.
Before we introduce the notations and dive into the construction --- for both the warmup as well as historical reasons --- 
we would like to introduce a cruder hierarhcical representation which we call it \emph{nesting relation tree} $\hat{T}$ in \Cref{sec:cactus-preprocessing} and \Cref{sec:nesting-relation-tree}.
Both warmup subsections are not strictly required in \Cref{sec:gabow-chain} and \Cref{sec:cactus-hierarchy} but they become useful for proving correctness in \Cref{sec:cactus-algorithm}.

\subsection{Warmup I: Assumptions and Preprocessing}\label{sec:cactus-preprocessing}

Let $r$ be the root we chose to fix throughout the algorithm.
As a technical reminder: from now on,
we will abuse the notation and identify
every cut $(X, V\setminus X)$ by the subset $X$ where $r\notin X$.
Any two cuts $X$ and $Y$ are nesting if either $X\subseteq Y$ or $Y\subseteq X$.
For any vertex $v\neq r$ let $X_v$ to be the minimal mincut of $v$ on $G$. For convenience we also define $X_r=V$.
Similarly, for each edge $e$, if there exists a minimal mincut for $e$ it is denoted by $X_e$; otherwise we set $X_e=V$.

\paragraph{Preprocessing: Merging Vertices with the Same Minimal Mincuts.}
First, without loss of generality we may assume that all vertices (except $r$) have distinct minimal mincuts.
This assumption can be achieved easily by merging vertices with the same minimal mincut: if two vertices $u$ and $v$ have the same minimal mincut, then any mincut does not separate $u$ and $v$. If a vertex $u$ does not have a minimal mincut, then any mincut does not separate $u$ and $r$. Therefore, all mincuts are preserved in the merged graph.

The preprocessing can be done in linear time deterministically: the algorithm first performs a bucket sort to all vertices' minimal mincut labels according to their sizes and breaking ties lexicographically.
Then, two vertices can be merged if and only if they are neighbors in the sorted order and they have identical cut labels. Notice that we may assume that two cuts are the same if and only if their labels are the same.

\subsection{Warmup II: Nesting Relation Tree}
\label{sec:nesting-relation-tree}

Readers may skip this subsection if the goal is to obtain only a high-level idea of the algorithm.

Let $v\neq r$ be a non-root vertex.
By the non-crossing property (\Cref{lem:crossing-mincuts}) of the mincuts, any two minimal mincuts of vertices are either nesting or disjoint.
With the preprocessing from \Cref{sec:cactus-preprocessing},
we are able to see that for all $v\neq r$, there exists a unique vertex $p_v\in V$ such that $X_v\subsetneq X_{p_v}$ and there exists no vertex $u\notin\{v, p\}$ with $X_v\subsetneq X_u\subsetneq X_{p_v}$.\footnote{In \cite{karger2009near} the authors call $X_{p_v}$ the second smallest minimal mincut of $v$.} 

Now, the nesting relation tree $\hat{T}$ can be defined by designating $p_v$ as the parent of $v$ for each non-root $v$.
It is straightforward to check that $\hat{T}$ is a well-defined tree, by observing that the nesting relations among all minimal mincuts of vertices are acyclic.
The following lemma states that $\hat{T}$ can be constructed efficiently.

\begin{lemma}\label{lem:cactus-compute-T-hat}
Given the labels to minimal mincuts of vertices and the tree packing $\T$, there is a deterministic algorithm that constructs $\hat{T}$ in $O(n|\T|)$ time.
\end{lemma}

\begin{proof}
Consider the collection of all minimal mincuts of vertices $\mathscr{X} := \{X_v\}_{v\in V}$.
For each vertex $v\neq r$, to find its parent $p_v$ on $\hat{T}$, it suffices to search for the second smallest sized set $\mathscr{X}$ that contains $v$ (the smallest one is $X_v$ by definition.)

This can be done by considering each tree separatedly.
Recall that each mincut has a label $(type, v, w, T)$.
For each tree $T\in\T$, we define $\mathscr{X}_T$ to be all minimal mincuts of vertices with its label referring to $T$.
For any vertex $v$, if we are able to obtain the smallest and the second smallest mincuts of $\mathscr{X}_T$ for every tree $T\in\T$, then we are able to obtain the parent $p_v$ by selecting the second smallest sized minimal mincuts among the $\le 2|\T|$ mincuts that were returned.

Now, observe that $\mathscr{X}_T$ is also a laminar set.
Thus, if we mark the $O(1)$ entry points for each mincut in $\mathscr{X}_T$ on the tree $T$,
a simple DFS can be applied to find  the first and the second smallest mincuts in $\mathscr{X}_T$ that contains any given vertex in $O(n)$ total time for each tree $T$. 
Hence, the total time required is $O(n|\T|)$.
\end{proof}

For any vertex $v\in V$, we define $\hat{T}_v$ to be the set of vertices within the subtree rooted at $v$.
The following lemma gives a neat observation to this nesting relation tree $\hat{T}$:

\begin{lemma}\label{lem:cactus-T-hat-subtree}
For all $v\in V$, $X_v = \hat{T}_v$.
\end{lemma}

\begin{proof}
Observe that $X_u\subseteq X_v$ if and only if $u$ is a descendant of $v$. Then the statement follows.
\end{proof}

With the help of the nesting relation tree $\hat{T}$, the following lemma enables us to categorize any global mincut:

\begin{lemma}\label{lem:categorize-all-mincuts}
Let $X$ be any (proper) mincut on $G$. Then, there is a set of vertices $S\subseteq V$ such that all vertices in $S$ have the same parent on $\hat{T}$ and $X = \bigcup_{u\in S} \hat{T}_u$. 
\end{lemma}

\begin{proof}
First, observe that any mincut $X$ satisfies $X=\cup_{u\in X} X_u$. To see this, if $X$ crosses with some minimal mincut $X_v$ of vertex $v$, then by \Cref{lem:crossing-mincuts}, both $X_v\setminus X$ and $X_v\cap X$ are mincuts. Since $v$ lies in one of $X_v\setminus X$ and $X_v\cap X$, we have found a smaller mincut that contains $v$, contradicts to the definition of $X_v$.
By \Cref{lem:cactus-T-hat-subtree}, we further have that  $X=\cup_{u\in X} \hat{T}_u$ so $X$ is union of some subtrees in $\hat{T}$.

Now we would like to show that the roots of these subtrees have the same parent in $\hat{T}$.
Let $S\subseteq X$ be comprised of all vertices $x\in X$ such that the parent of $x$ is not in $X$. Suppose there are two vertices $x, y\in S$ so that their parent vertices $p$ and $q$ are different.
Since $p\neq q$, without loss of generality we may assume that $p\notin \hat{T}_{q}$
Now, $X$ must cross with $X_q(=\hat{T}_q)$ since $x\notin \hat{T}_q$ but $y\in \hat{T}_q$. In this case, one of $X_q\cap X$ or $X_q\setminus X$ is a smaller mincut that contains $q$, contradicting the definition of $X_q$.
Hence, all vertices in $S$ has the same parent and the statement is true.
\end{proof}

The next \Cref{sec:gabow-chain} gives important characterization to global mincuts, which leads to the actual cactus construction algorithm.
\footnote{
\cite{karger2009near} stopped at the above \Cref{lem:categorize-all-mincuts} and showed that by partitioning all mincuts according to the common parent vertex,
the cactus can be constructed layer-by-layer.
However, the correctness proof to the remaining part of \cite{karger2009near} dealing with one-layer cactus construction is not explicitly stated.
In the next subsection, we fill in the missing proof (and slightly simplify their algorithm) for Karger and Panigrahi's algorithm.}

\subsection{Chains and Chain Certificates}
\label{sec:gabow-chain}

The goal of this section is to show that every global mincut has a ``simple certificate'' corresponding to it. This will be  crucial for defining the hierarchical representation of all mincuts, which will be given in \Cref{sec:cactus-hierarchy}.
In fact, a generalization of the fact was already shown by Gabow \cite{gabow2016minset} but we give an alternative (and arguably simpler) presentation of this characterization using terminology specific to mincuts.

\begin{definition}[Mincut Certificates]
Let $X\subseteq V$ be any set of vertices. We say $X$ has a \emph{vertex certificate} if there is a vertex $v$ so that $X=X_v$.
Similarly, we say $X$ has an \emph{edge certificate} if $X$ is the minimal mincut for some edge $e$.
\end{definition}

By the end of this section, we will prove in \Cref{lem:certificate char} that every mincut must either has a vertex certificate, an edge certificate, or a \emph{chain certificate}. In order to understand the last object, we first give the definitions of chains and chain certificates below.   

\begin{definition}[inspired by \cite{gabow2016minset}, page 33--35]
\label{def:chain}
A \emph{chain} is a sequence of disjoint non-empty vertex subsets $(C_0, C_1, \ldots, C_\ell)$ where $\ell \ge 1$ and is defined recursively:
\begin{itemize}[itemsep=0pt]
    \item[1.] For each $i$, $C_i\subseteq V$ has either a vertex certificate, an edge certificate, or a \emph{chain certificate}.
    \item[2.] For each $0\le i < \ell$, $C_i\cup C_{i+1}$ has an edge certificate.
\end{itemize}
We say that a subset of vertices $X\subseteq V$ has a \emph{chain certificate} if there exists a chain $(C_0, C_1, \ldots, C_\ell)$ such that $X=\cup_{i=0}^{\ell} C_i$.
A \emph{subchain} is a consecutive subsequence of a chain.
A \emph{maximal chain} is a chain that is not a subchain of any longer chain.
\end{definition}

It is very important to note that \Cref{def:chain} is well-defined, in the sense that once all subsets with vertex or edge certificates are fixed, chains and maximal chains will be unambiguously defined.
Now, we explore some useful properties of chains. These properties can all be derived from the crossing lemma (\Cref{lem:crossing-mincuts}).

\begin{lemma}[basic properties of a chain]\label{lem:chain-property}
Let $(C_0, C_1, \ldots, C_\ell)$ be a chain and let $X=C_0\cup C_1\cup \cdots \cup C_\ell$. Then the following properties are true:
\begin{enumerate}[itemsep=0pt]
    \item[\namedlabel{prop:1}{(1)}] $X$ is either a mincut or the vertex set $V$.
    \item[\namedlabel{prop:chain-edges-allow-to-leave-only-at-ends}{(2)}] For any edge $e$ leaving $X$ (i.e., $e$ connects $X$ and $V\setminus X$), then $e$ has an endpoint in either $C_0$ or $C_\ell$.
    \item[\namedlabel{prop:chain-neighbor-degree}{(3)}] For any $0\le i < \ell$, $\C(C_i, C_{i+1})=\lambda / 2$; also, $\C(C_0, V\setminus X)=\C(C_\ell, V\setminus X)=\lambda / 2$. In particular, for any non-neighbor indices $i$ and $j$ such that $|i-j| > 1$ we have $\C(C_i, C_j)=0$.
    \item[\namedlabel{prop:components-are-maximal-chain}{(4)}] (Lemma 4.6 in \cite{gabow2016minset}) For any $i$, if $C_i$ has a chain certificate $(C'_0, C'_1, \ldots, C'_{\ell'})$, then this certificate must be a maximal chain.
    \item[\namedlabel{prop:extend-a-component-cannot-cross}{(5)}] For any edge $e$ leaving $C_0$ (resp. $C_\ell$) to $V\setminus X$, the minimal mincut of $e$ either contains  $C_0$ (resp. $C_\ell$) but nothing from $X\setminus C_0$ (resp. $X\setminus C_\ell$), or contains the entire $X$.
\end{enumerate}
\end{lemma}

\begin{proof} We prove each property one by one below.
\begin{description}
    \item[Property \ref{prop:1}.] This can be done by induction on $|X|$ when $|X|<n$: we know that $C_\ell$ has a certificate so it is a mincut. Now, by the fact that $X\setminus C_\ell$ cross with $C_{\ell-1}\cup C_{\ell}$, by \Cref{lem:crossing-mincuts} we know that $X$ is a mincut.
    \item[Property \ref{prop:chain-edges-allow-to-leave-only-at-ends}.] Fix $i$ with $1\le i \le \ell-1$. 
    Since any subchain is a chain, by property~\ref{prop:1} we have $A:=C_0\cup\cdots \cup C_i$ and $B:=C_i\cup C_{i+1}\cup \cdots \cup C_\ell$ are two mincuts. By \Cref{lem:crossing-mincuts}, there is no edge going from $A\cap B=C_i$ to the outside $V\setminus X$.
    \item[Property \ref{prop:chain-neighbor-degree}.] For any $0\le i < \ell$, we focus on the term $\C(C_i, C_{i+1})$ by the following: let $A=V\setminus (C_i\cup C_{i+1}))$, since all of $C_i$, $C_{i+1}$, and $C_i\cup C_{i+1}$ are mincuts, we know that 
    any two terms of $\C(A, C_i)$, $\C(C_i, C_{i+1})$, and  $\C(C_{i+1}, A)$ add up to the value $\lambda$. Thus, each of them must be exactly $\lambda / 2$. To show that $\C(C_0, V\setminus X)=\lambda/2$ we apply the same argument to the three mincuts $C_0$, $C_1\cup\cdots \cup C_\ell$ and $X$. Similar argument for $\C(C_\ell, V\setminus X)=\lambda/2$.
    
    To prove the second statement, we assume $i+1<j$ and use \Cref{lem:crossing-mincuts} on mincuts $A:=C_i\cup \cdots \cup C_{j-1}$ and $B:=C_{i+1}\cup \cdots\cup C_j$. Since $|i-j|>1$, we know that $A$ crosses with $B$ and hence $\C(C_i, C_j)=\C(A\setminus B, B\setminus A) = 0$.
    \item[Property \ref{prop:components-are-maximal-chain}.]
    Without loss of generality assume $i < \ell$ (otherwise we simply reverse the chain.)
    Assume for the contradiction that the chain certificate of $C_i$ is not maximal.
    Without loss of generality there exists an edge $e$ such that $X_e = C'_{\ell'}\cup C'_{\ell'+1}$. Moreover, since $C_i$ is part of a chain, there exists an edge $f$ so that $X_f = C_{i}\cup C_{i+1}$.
    By property~\ref{prop:chain-edges-allow-to-leave-only-at-ends}, we know that one of $f$'s endpoints must be in $C'_0$ (i.e., cannot be leaving from $C'_{\ell'}$).
    By property~\ref{prop:chain-neighbor-degree}, the other endpoint of $f$ should not occur in $C'_{\ell'+1}$. So both endpoints of $f$ do not occur in $X_e$.
    
    If $X_e$ and $X_f$ cross, by the crossing lemma (\Cref{lem:crossing-mincuts}), $X_f\setminus X_e$ is a smaller mincut that contains $f$, a contradiction.
    
    If $X_e$ and $X_f$ do not cross, then $X_e\subseteq X_f$.
    Notice that now $Y:=C_i\cup X_e$ is a mincut that crosses with $C_{i+1}$.
    By the crossing lemma again (\Cref{lem:crossing-mincuts}) we know that $\C(Y\setminus C_{i+1}, C_{i+1}\setminus Y)=0$. However, $f$ is an edge that goes from $C_i\setminus X_e$ ($\subseteq Y\setminus C_{i+1}$) to $C_{i+1}\setminus X_e$ ($\subseteq C_{i+1}\setminus Y$), a contradiction too.
    \item[Property \ref{prop:extend-a-component-cannot-cross}.]
    We first show that $C_0\subseteq X_e$: suppose not, $C_0$ cross with $X_e$ (since $e$ has an endpoint in $C_0$). Let $f$ be an edge so that $X_f=C_0\cup C_1$. Notice that $X_e$ and $X_f$ cross and so $X_f\setminus X_e$ is a smaller mincut. Observe that $X_f \setminus X_e$
    contains both endpoints of $f$ because, by \Cref{lem:crossing-mincuts}, there is no edge from $X_e\cap C_0$ to $C_1\subseteq (V\setminus X_e\cup C_0)$. This contradicts to the fact that $X_f$ is the minimal mincut of $f$.
    
    Now that $C_0\subseteq X_e$, we assume for the contradiction that $X_e$ contains strictly more than $C_0$ but not entire $X$.
    Then, $X_e$ cross with $A:=C_1\cup C_2\cup\cdots\cup C_\ell$.
    By \Cref{lem:crossing-mincuts} $X_e\setminus A$ is also a mincut.
    Now $X_e\setminus A$ is a smaller mincut that contains $e$, a contradiction too.
    
    Therefore, $X_e$ contains either entire $X$ or just $C_0$ but none of vertices from $X\setminus C_0$.\qedhere
\end{description}
\end{proof}

As we will see, each chain actually corresponds to a further-from-root part of a cycle on a cactus.

Now, Gabow's characterization of mincuts shows that every mincut has a certificate:

\begin{lemma}[\cite{gabow2016minset}, Lemma 4.4 and Lemma 4.5]\label{lem:certificate char}
Every mincut on $G$ has either a vertex certificate, an edge certificate, or a chain certificate.
\end{lemma}

\begin{proof}
Assume $G$ is connected (so $\lambda > 0$). Let $X$ be a mincut (so the induced subgraph $G[X]$ is connected).
We prove by induction on the size of $X$.

\noindent\textbf{Base Case.} When $|X|=1$, $X$ trivially has a vertex certificate.

\noindent\textbf{Inductive Case.} Now suppose $|X|>1$. If $X$ has a vertex/edge certificate then we are done.

Suppose $X$ has no vertex/edge certificate.
Consider the collection $\mathcal{X}$ that consists of all maximal proper subsets of $X$ that are mincuts.
Then, since $X_v\subsetneq X$ for all $v\in X$, we know that $\mathcal{X}\neq \emptyset$.
Take any $A\in\mathcal{X}$.
By induction hypothesis, $A$ has a certificate denoted as $(C_0, \ldots, C_{\ell})$.
For convenience, if $A$ has a vertex certificate or an edge certificate, we use the same notation where $C_0:=A$ and $\ell=0$.

Finally, we claim that we can actually attach the entire $X\setminus A$ to either end of the certificate, forming an authentic chain $(C_0, \ldots, C_\ell, X\setminus A)$ or $(X\setminus A, C_0, \ldots, C_\ell)$ and conclude the proof.
Since $G[X]$ is connected, there is at least one edge $e\in E(A, X\setminus A)$.
Let $X_e$ to be the minimal mincut of $e$.
We know that $X_e\subsetneq X$ because $X$ has no edge certificate.
Since  $X_e \cap A\neq \emptyset$ but $X_e$ is not a superset of $A$ (since $A$ is maximal), we know that $X_e$ cross $A$. By \Cref{lem:crossing-mincuts}, both $X_e\setminus A$ and $X_e\cup A$ are mincuts.
Now, we notice that $X_e\cup A$ is a superset of $A$. This implies that $X_e\cup A=X$ and $X_e\setminus A=X\setminus A$ is a mincut.

By property~\ref{prop:chain-edges-allow-to-leave-only-at-ends} of \Cref{lem:chain-property}, $e$ must leave from $C_0$ or $C_\ell$. By property~\ref{prop:extend-a-component-cannot-cross} of \Cref{lem:chain-property}, with the clue $X_e\neq X$, we know that either $X_e=(X\setminus A)\cup C_0$ or $X_e=(X\setminus A)\cup  C_{\ell}$. Hence, by definition, one of $(X\setminus A, C_0, \ldots, C_\ell)$ or $(C_0, \ldots, C_\ell, X\setminus A)$ is a chain certificate of $X$.
\end{proof}

The following \Cref{lem:chain-uniqueness} shows that chain certificates of a mincut are basically unique. This is helpful in the following sense: 
suppose we have extended the chain certificate of a mincut $A$ to some longer chain (let $B\supset A$ be the associated mincut).
Then, we do not need to store $A$ in the memory because the mincut $A$ is ``preserved'' (as a subchain) in any chain certificate of $B$.

\begin{lemma}[Briefly mentioned in \cite{gabow2016minset}, page 36]\label{lem:chain-uniqueness}
Let $X$ be a mincut that has a chain certificate $(C_0, C_1, \ldots, C_\ell)$. Then this certificate is unique up to reversing the chain.
\end{lemma}

\begin{proof}
Let $e=(u, v)$ be any edge from the boundary $E(V\setminus X, C_0)$ with $v\in C_0$. (The existence of $e$ is guaranteed by property~\ref{prop:chain-edges-allow-to-leave-only-at-ends} of \Cref{lem:chain-property}.)
Let $(C'_0, C'_1, \ldots, C'_{\ell'})$ be any maximal chain certificate of $X$. Without loss of generality we have $v\in C'_0$ (reverse the chain if necessary).

Now, let $i$ be the smallest index such that $C_i\neq C'_i$. It is easy to see that $i=0$: otherwise there are two edges $e$ and $f$ such that $X_e=C_{i-1}\cup C_i$ and $X_f=C_{i-1}\cup C'_i$. 
By property~\ref{prop:chain-neighbor-degree}, the endpoints of both $e$ and $f$ belongs to $C_i\cap C'_i$. We obtain a contradiction as $C_{i-1}\cup (C_i\cap C'_i)$ is a smaller mincut for either $e$ or $f$ (or both). Thus, $C_0\neq C'_0$ but $C_0\cap C'_0\neq\emptyset$.
Depending on whether or not $C_0$ crosses with $C'_0$, there are two cases now:

\paragraph{Case 1:} If $C_0$ crosses $C'_0$, this implies a contradiction to \Cref{lem:crossing-mincuts} because $\C(C_0\cap C'_0, V\setminus (C_0\cap C'_0))=0$ but the edge $e$ connects $C_0\cap C'_0$ and $V\setminus X\subseteq V\setminus (C_0\cap C'_0)$. Therefore $C_i=C'_i$ for all $i$ and the chain is unique.

\paragraph{Case 2:} If $C_0$ and $C'_0$ do not cross, then without loss of generality we may assume $C_0\subset C'_0$. %
Suppose there exists an index $i>0$ such that $C'_0$ crosses with $C_i$ and let $i$ to be the smallest one among all such indices.
Then, for each index $j$ such that $0\le j < i$, we know that either $C_j\subsetneq C'_0$ or $C_j\cap C'_0=\emptyset$.
However the latter case is impossible: notice that $C_0\subseteq C'_0$ and $C_i\cap C'_0$, by property \ref{prop:chain-neighbor-degree}, $C_0$ and $C_i\cap C'_0$ are not connected in $C'_0$. By property \ref{prop:1} $C'_0$ should be a mincut and hence must be connected, a contradiction.

Thus, $C'_0$ should contain all the $C_j$ for $j < i$, and there exists an edge $e$ such that $X_e = C_{i-1}\cup C_i$. First observe that $C'_0$ doesn't contain both endpoints of $e$, otherwise $C'_0\cap (C_{i-1}\cup C_i)$ will be a smaller mincut containing $e$ for both cases that $C'_0$ crossing with $C_{i-1}\cup C_i$ and $C'_0\subset C_{i-1}\cup C_i$. Then  the edge $e$ connects $C'_0\setminus C_i$ and $C_i\setminus C'_0$ but this contradicts to \Cref{lem:crossing-mincuts} 
for $\C(C'_0\setminus C_i, C_i\setminus C'_0) = 0$ since $C'_0$ crosses $C_i$. Therefore $C'_0$ must be a subchain $(C_0, C_1,\ldots,C_k)$ of the chain $(C_0, C_1,\ldots,C_\ell)$, but this contradicts to property \ref{prop:components-are-maximal-chain} of \Cref{lem:chain-property} since it is not a maximal chain.
\end{proof}

\subsection{Reducing Cactus to Hierarchical Representation of Global Mincuts}
\label{sec:cactus-hierarchy}

The goal of this section is to reduce the problem of cactus construction to constructing a \emph{hierarchical representation} $H$ of mincuts based on chain certificates defined in \Cref{sec:gabow-chain}.\footnote{Gabow~\cite{gabow2016minset} defined
special directed graphs based on partial order sets called \emph{chain-trees}.
However, the chain-trees are not uniquely defined.
Gabow showed that a specific chain-tree can be algorithmically created based on a particular laminar collection of mincuts.
Here we define the hierarchical representation $H$ based on exactly the same collection as Gabow's.
There are two benefits introducing $H$: (1) the representation itself is based purely on the structural property of $G$, not algorithmically, and 
(2) with the preprocessing mentioned in \Cref{sec:cactus-preprocessing}, the representation is unique up to reversing the chains in the chain certificates.
}

\begin{lemma}[\cite{gabow2016minset}, Lemma 4.6 (i)]\label{lem:laminar}
Consider the collection $\collection$ of all mincuts $X$ such that either all chain certificates of $X$ are maximal, or $X$ has no chain certificate.
Then $\collection$ is laminar.\footnote{The collection defined corresponds to $\mathcal{F}^-$ in  \cite{gabow2016minset}.}
\end{lemma}

\begin{proof}
We first show that for any two mincuts $A$ and $B$ that cross each other, then
both $A$ and $B$ have a chain certificate.
Suppose $A$ does not have a chain certificate, by \Cref{lem:certificate char} we know that $A$ has a vertex certificate $(A=X_v)$ or an edge certificate $(A=X_e)$. If $A=X_v$ for some $v\in V$ then by the crossing lemma either $A\setminus B$ or $A\cap B$ is a smaller mincut that contains $v$, a contradiction. If $A=X_e$ then $e$ must have one endpoint in $A\setminus B$ and another in $A \cap B$, otherwise one of the two sets contains $e$, which contradicts minimality of $X_e$. But this implies that $(A\setminus B, A\cap B)$ is a chain certificate for $A$, which is also a contradiction. The same proof applies to $B$.

Now, we show that if two mincuts $A$ and $B$ that cross each other, then none of $A$ and $B$ has a maximal chain certificate. 
Consider any edge $e\in E(A\cap B, B\setminus A)$.
Since both endpoint of $e$ are in $B$, we know that $X_e\subseteq B$. Now, since $A$ has a chain certificate, and $X_e$ does not contain entire $A$, by property~\ref{prop:extend-a-component-cannot-cross} from \Cref{lem:chain-property} we know that $e$ is an edge that extends the chain certificate of $A$. Hence, there is a chain certificate of $A$ that is not maximal. Same proof for $B$.

To conclude the statement, we notice that if there are two mincuts $A$ and $B$ crossing each other, then both $A$ and $B$ does not appear in $\collection$. So $\collection$ must be laminar.
\end{proof}

As a sanity check we note that a mincut with a vertex certificate has no chain certificate, so $X_u\in\collection$ for all $u\in V$.
Since $\collection$ is laminar on $V$, we know that $|\collection| = O(n)$ and the total length of chain certificates from mincuts in $\collection$ is also $O(n)$.

\paragraph{The Hierarchical Representation.} Let $\collection$ be the collection described in \Cref{lem:laminar}.
The collection $\collection$ naturally defines a \emph{hierarchy} tree $H$:
for each mincut $A\in\collection$ there is a node $v_A$ in $H$. Moreover,
the children of $v_A$ in $H$ are all nodes $v_B$ such that $B\in\collection$ is the maximal proper subset of $A$.

According to the definition of $H$ and \Cref{lem:chain-uniqueness}, the hierarchy $H$ is unique, but the chain certificates are unique up to reversion. 
In particular,
\Cref{lem:chain-uniqueness} implies that every mincut $X$ on the graph $G$ can be ``captured'' by $H$: if $X$ does not have a chain certificate, then $X\in\collection$. Otherwise, $X$ has a chain certificate. By extending this chain certificate to a maximal chain, we know that there exists a mincut $Y\in\collection$ such that the chain certificate of $X$ occurs as a subchain to $Y$'s maximal chain certificate.

We finish this subsection by showing that a hierarchy $H$ can be easily turned into a cactus.
Hence, once a hierarchy $H$ is formed with all certificates given, a cactus representation of graph $G$ can be constructed in linear time.

\begin{lemma}[\cite{gabow2016minset}, Section 4.5]\label{lem:hierarchy-to-cactus}
Given a graph $G=(V, E)$ and its corresponding hierarchy $H$ with certificates in $O(n)$ total size,
a cactus representation of $G$ can be constructed in $O(n)$ time. 
\end{lemma}

\begin{proof}
A cactus graph $P$ can be constructed from the hierarchy $H$ as follows. For each node $v_A$ on $H$ where the corresponding mincut $A\in\collection$ has a (maximal) chain certificate $(C_0, C_1, \ldots, C_\ell)$.
By property~\ref{prop:components-are-maximal-chain} of \Cref{lem:chain-property}, all parts $C_i$ either have a maximal chain certificate or a vertex/edge certificate. 
Any superset of $C_i$ that is proper in $A$ has a chain representation but they are always extendable.
Thus, for each $C_i$ there is a corresponding node $v_{C_i}$ in $H$ and they are children of $v_A$.
We replace this star $\{(v_A, v_{C_i})\}_{i=0}^\ell$ by a cycle of length $\ell+1$: $(v_A, v_{C_0}, v_{C_1}, \ldots, v_{C_\ell})$.

After all the replacement are done,
it is straightforward to verify that $P$ is a cactus.
Now we assign vertices $V$ to nodes in $P$. Since $X_u\in\collection$ for all $u\in V$, by simply assign each vertex $u\in V$ to the node $v_{X_u}$ then we are done.

All mincuts are preserved: a mincut with a chain certificate can be found by cutting two edges from the replaced cycle. A mincut $A$ without a chain certificate can be found by cutting the edge from $v_A$ with its parent on $P$. On the other hand, cutting a bridge or a pair of edges in the same cycle corresponds to a mincut too.
\end{proof}

Now the task of constructing a cactus representation reduces to computing a hierarchy representation $H$.

\subsection{Constructing a Hierarchical Representation}
\label{sec:cactus-algorithm}
\label{sec:solve-one-layer-cactus}
\label{sec:one level reduc}

In this section we describe an algorithm that constructs the hierarchy $H$  defined in \Cref{sec:cactus-hierarchy}.\footnote{Our algorithm is simpler than Karger and Panigrahi's algorithm because we do not reduce the problem into one-layer cactus construction problems.}

Assume that the graph has been preprocessed as described in \Cref{sec:cactus-preprocessing} such that every vertex has a distinct minimal mincut.
The algorithm constructs the hierarchy by processing all vertex/edge-certificated mincuts in the non-decreasing order of their sizes.

Two variables are introduced explicitly: a (partial) collection $\collection$ and a (partial) hierarchy forest $H$.
$\collection$ is initialized as an empty set and $H$ is initialized as an empty forest.
At any moment, the algorithm maintains a collection $\collection$ of 
disjoint mincuts with corresponding certificate.
In the meantime, the algorithm maintains a hierarchy forest $H$ with the invariant such that
there is always a bijection between $\collection$ and the roots of all trees in $H$.
Throughout processing the mincuts,
two subsets in $\collection$ may be merged (so the corresponding trees in $H$ may be merged at the root.)
Once a mincut $A\in\collection$ is removed, we guarantee that some superset $B$ of $A$ is added to $\collection$ and $v_A$ becomes a child of $v_B$ in $H$.

Before we describe the steps of the algorithm, we state the most important property that leads to the correctness of the algorithm.

\begin{lemma}[Inductive Correctness Guarantee]\label{inv:cactus-correctness-invariant}
Upon processing a mincut of size $t$, any mincut of size strictly less than $t$ can be uniquely ``represented'' in $H$. That is, let $X$ be a mincut with $|X|<t$. If $X$ has a chain certificate $C$, then there is a unique node $v_A$ in $H$ with a chain certificate containing $C$ as a subchain.
If $X$ does not have a chain certificate, then $X$ has a corresponding node $v_X\in H$.
\end{lemma}

This property explains the validity of certain steps in our algorithm.
As you can see, by \Cref{inv:cactus-correctness-invariant}, at the end the algorithm returns a correct hierarchy representation $H$ defined in \Cref{sec:cactus-hierarchy}.
Now we describe this high-level algorithm.

Let $\cL$ be the list of minimal mincuts of vertices and edges. The mincuts in $\cL$ is sorted by their size in increasing order. The algorithm processes mincuts in $\cL$ one by one.  
Let $X \in \cL$ be the current processing minimal mincut.
If there is already a mincut $A\in\collection$ such that $X\subseteq A$, then the algorithm does nothing. 
Now we assume the otherwise: $X$ is not contained in any mincut in $\collection$.
$X$ could be a minimal mincut of a vertex $u\in V$, or a minimal mincut of an edge $e\in E$.

Suppose that $X=X_u$ has a vertex certificate.
Using the crossing lemma (\Cref{lem:crossing-mincuts}), we deduce that the minimal mincut $X_u$ does not cross with any mincut.
In particular, $X_u$ does not cross any mincut in the current $\collection$.
In this case, we remove any mincut that is a subset of $X_u$ from $\collection$ and add $X_u$ to $\collection$. The hierarchy forest $H$ is updated accordingly, by creating a node $v_{X_u}$ and assign the nodes $v_A$ to be its children for all $A\subseteq X_u$ that was removed from $\collection$.

Suppose that $X=X_e$ has an edge certificate $e=(u, v)\in E$.
The algorithm checks if this edge helps extending or creating a chain.
We observe that right now $u$ and $v$ must belong to different mincuts in $\collection$, otherwise there is already a mincut $A\in\collection$ containing $X$.
Let the mincuts $A, B\in\collection$ such that $u\in A$ and $v\in B$.
If $X_e$ crosses with one of the mincut, say $A$, then we must have $X\setminus A$ be a mincut.
By \Cref{inv:cactus-correctness-invariant}, since $|X\setminus A| < |X|$, we know that the mincut $X\setminus A$ must have already been represented in the current hierarchy $H$,
which implies that there exists a mincut in the current collection $\collection$ that contains $X\setminus A$.
Since $v\in X\setminus A$ and by the assumption $v\in B\in\collection$,
we know that $X\setminus A\subseteq B$.

Therefore, only three cases can occur between $X$ and $A\cup B$: either $X=A\cup B$, $X\subsetneq A\cup B$, or $X\supsetneq A\cup B$.
In the first case where $X=A\cup B$, we know that a new chain is formed. We create a new node $v_X$ and set $v_A$ and $v_B$ be its children in $H$. The certificate of $v_X$ will be a chain certificate $(A, B)$.
In the second case where $X\subsetneq A\cup B$, we know that some chain can be extended (or possibly, two chains are concatenated.)
Without loss of generality we assume that $X$ cross with $A$.
Now there will be two sub-cases, either $X\setminus A=B$ or $X\setminus A\subsetneq B$. If $X\setminus A=B$ then
on the hierarchy $H$ we rename $v_A$ to be $v_{A\cup B}$ and then put $v_B$ as a new child of $v_{A\cup B}$. If $X\setminus A\subsetneq B$, we know that $e$ concatenates the two chains of $A$ and $B$. On the hierarchy $H$ it suffices to merge the two trees rooted at $v_A$ and $v_B$, and update the certificate of $v_{A\cup B}$.
In the third case where $X\supsetneq A\cup B$, similar to the vertex case,
the algorithm removes all mincuts that are contained in $X$. The hierarchy $H$ is updated accordingly.

We summarize this high-level algorithm in \Cref{alg:cactus-high-level}.

\begin{algorithm}[H]
\algnotext{EndIf}
\algnotext{EndFor}
    \caption{High-Level Approach for Constructing a Hierarchy Representation}
    \label{alg:cactus-high-level}
    \begin{algorithmic}[1]
    \Require{A graph $G=(V, E)$, tree packing $\T$, labels of minimal mincuts of all vertices and edges.}
    \Ensure{A hierarchy tree $H=(V_H, E_H)$. Every node on $H$ has a certificate.}
        \State $\cL \gets$ the list of minimal mincuts of all vertices and edges, sorted by their sizes.\label{line:cactus-sort-all-minimal-mincuts}
        \State $\collection\gets \emptyset$. \Comment{A disjoint collection of mincuts.}
        \State $H\gets \emptyset$. \Comment{A hierarchy forest with certificates.}
        \For{each vertex/edge minimal mincut $X \in \cL$} 
            \If{there does not exist $A\in\collection$ such that $X\subseteq A$}\label{line:cactus-if-already-done}
            \If{$X$ has a vertex certificate $u\in V$}
                \State $\texttt{AddNestingSuperset}(X_u)$.\label{line:add-nesting-superset-vertex}
            \Else \ \ (now {$X$ has an edge certificate $e=(u, v)\in E$})
                \State Let $A, B\in \collection$ such that $u\in A$ and $v\in B$.
                \If{$X=A\cup B$}\label{line:cactus-if-edge-cases}
                \State $\texttt{AddNewChain}(X_e, A, B)$.\label{line:add-new-chain} \Comment{Create a new chain $(A, B)$.}
                \ElsIf{$X\subsetneq A\cup B$}
                    \State Swap the role of $A$ and $B$ so that $A$ cross with $X$.\label{line:cactus-swap-role-so-A-cross-X}
                    \If{$X\setminus A=B$}\label{line:cactus-sub-cases}
                        \State $\texttt{ExtendChain}(X_e, A, B)$.
                        \label{line:extend-chain} \Comment{Extend $A$'s chain by adding $B$ at the end.}
                    \Else \ \ (now $X\setminus A \subsetneq B$)
                        \State $\texttt{ConcatChains}(X_e, A, B)$.\label{line:concatenate-chains}\Comment{Concat $A$ and $B$'s chain certificates via $X_e$.}
                    \EndIf
                \Else \ \ (now $X\supsetneq A\cup B$)
                \State $\texttt{AddNestingSuperset}(X_e)$. \label{line:add-nesting-superset-edge}
                \EndIf
            \EndIf
            \EndIf
        \EndFor
    \end{algorithmic}
\end{algorithm}

\paragraph{Correctness Proof.} We end this subsection by proving \Cref{inv:cactus-correctness-invariant}, which implies that \Cref{alg:cactus-high-level} does produce a correct hierarchy $H$ after processing all minimal mincuts of vertices and edges.

\begin{proof}[Proof of \Cref{inv:cactus-correctness-invariant}.]
Let's apply mathematical induction on $t$ and 
let $A$ be a mincut on $G$ of size $|A|<t$.
First, since a chain certificate can only be extended or concatenated,
it is straightforward to see that if $A$ is added to $\collection$, $A$ will be ``represented'' in any future moment.%
Notice that whenever $A$ has a vertex certificate or an edge certificate, either $A$ will be added to $\collection$ or is used for extending/concatenating a chain.

Now, assume that $A$ has a chain certificate $(C_0, C_1, \ldots, C_\ell)$ with $\ell \ge 1$.
Since $|A|<t$, for any $0\le i < \ell$ we know that $|C_i\cup C_{i+1}| < t$ as well.
Since $C_i\cup C_{i+1}$ has an edge certificate, we know that $C_i\cup C_{i+1}$ must have been processed already for all $i$.
If $\ell=1$, then Line~\ref{line:add-new-chain} correctly construct a chain.
If $\ell > 1$, then the mincuts $A_l := C_0\cup\cdots \cup C_{\ell-1}$ and $A_r := C_1\cup \cdots \cup C_{\ell}$ can be found uniquely in some node (say $v_l$ and $v_r$ respectively) on $H$, by the induction hypothesis of \Cref{inv:cactus-correctness-invariant}.
Since the node that represents $C_1$ can also be uniquely found as a child of both $v_l$ and $v_r$, we conclude that $v_l=v_r$. Therefore, $A$ can be found uniquely in $v=v_l=v_r$ too.
\end{proof}

\subsection{Efficient Implementation of \Cref{alg:cactus-high-level}}

Once the high-level idea is confirmed, the remaining parts of the implementation become relatively easier tasks.
There may be different ways to implement \Cref{alg:cactus-high-level}, and we give one of them in this subsection.

For $\collection$ the algorithm maintains an additional disjoint set data structure (with \textsc{Union} and \textsc{Find} operations).
For any mincut $A \in \collection$, we store (1) its size $|A|$, (2) one vertex $v\in A$, and (3) certificates of $A$.
If there are multiple certificates available for the same mincut, we store one certificate of each kind: vertex, edge, and chain.
For a chain certificate $(C_0, C_1, \ldots, C_{\ell})$, we assume that a doubly linked list of edges $(e_1, e_2, \ldots, e_{\ell})$ is stored in the memory where $X_{e_i}=C_{i-1}\cup C_i$ for all $1\le i\le \ell$. That said, the operations to chains (e.g., \texttt{AddNewChain} in Line~\ref{line:add-new-chain}, \texttt{ExtendChain} in Line~\ref{line:extend-chain}, and \texttt{ConcatChains} in Line~\ref{line:concatenate-chains}) can be implemented in a straightforward way in $O(1)$ time). 

\paragraph{Containment Queries.}

In \Cref{alg:cactus-high-level}, the algorithm is often required to test whether two given mincuts $A$ and $B$ satisfies $A\subseteq B$.
(Specifically, this operation is used to implement Line~\ref{line:cactus-if-already-done}, Line~\ref{line:cactus-if-edge-cases}, and Line~\ref{line:cactus-sub-cases}.)
This test is denoted by Karger and Panigrahi~\cite{karger2009near} as the \emph{containment query}: \textsc{Containment}$(A, B)$ returns \textsf{true} if and only if $A\subseteq B$.
In \cite{karger2009near} the authors 
 merely mentioned that the containment queries can be answered in $O(1)$ by 
 answering LCA queries in the corresponding tree.
 The authors did not describe an algorithm that answers containment queries --- it becomes highly non-trivial when $A$ or $B$ has a chain certificate.
 Fortunately, thanks to the crossing lemma, 
 most of the containment queries 
 in the high-level \Cref{alg:cactus-high-level}
 can be implemented by simply checking and comparing the sizes of the mincuts. In below, we describe the detailed implementations line by line.

\begin{description}
    \item[Line \ref{line:cactus-sort-all-minimal-mincuts}.]
    First of all, Line~\ref{line:cactus-sort-all-minimal-mincuts} in \Cref{alg:cactus-high-level} can be done efficiently in $O(m+n|\T|)$ time, by computing the sizes of the mincuts in $O(n|\T|)$ time and performing a bucket sort in $O(m+n)$ time.
    Notice that in order to correctly implement Line~\ref{line:add-nesting-superset-edge}, we require that the same minimal mincuts are listed together.
    This can be achieved by, e.g., breaking ties using the lexicographical order of the cut labels.
    Moreover, for the same mincut we process vertices first then the edges.
    
    \item[Line~\ref{line:cactus-if-already-done}.]
Since all mincuts in $\collection$ are disjoint and they are only replaced by supersets, it suffices to use a standard disjoint set data structure supporting membership queries.
In particular, $\textsc{Find}(x)$ returns the mincut in $\collection$ that contains $x$, or $\perp$ if such mincut does not exist.

To implement Line~\ref{line:cactus-if-already-done}, if $X=X_u$ is a minimal mincut of vertex $u$, then we know that $\textsc{Find}(u)=\perp$, and that the \textbf{if} statement is always evaluated to \textsf{true}.

If $X=X_e$ is a minimal mincut of an edge $e=(u, v)$. We observe that there exists $A\in\collection$ that contains $X$ if and only if $\textsc{Find}(u)=\textsc{Find}(v)$.

\item[Line~\ref{line:cactus-if-edge-cases}.]
Let $X=X_e$ with $e=(u, v)$. Suppose now that 
$\textsc{Find}(u)\neq \textsc{Find}(v)$ and there are two mincuts $A, B\in\collection$ with $u\in A$ and $v\in B$.
There are only three cases to distinguish: $X=A\cup B$, $X\subsetneq  A\cup B$, and $X\supsetneq A\cup B$.
Since $A\cap B=\emptyset$, it suffices to compare the size $|X|$ with $|A|+|B|$.

\item[Line~\ref{line:cactus-swap-role-so-A-cross-X}.]
To test whether $A\subseteq X$, we utilize the cut label of $X$ and the certificate of $A$.
Let $(type, v, w, T)$ be the cut label of $X$. 
Using a prebuilt data structure on $T$ it is easy to check whether a vertex belongs to $X$ in $O(1)$ time.

If $A=X_u$ has a vertex certificate, we know that $A\subseteq X$ for sure by definition of minimal mincut of $u$.
If $A=X_{e'}$ has an edge certificate $e'=(u', v')$, then definition of $X_{e'}$ we know that $A\subseteq X$ if and only if $u'\in X$ and $v'\in X$, and this can be tested in $O(1)$ time.
If $A$ has a chain certificate $(C_0, C_1, \ldots, C_\ell)$, we know that
by
property~\ref{prop:extend-a-component-cannot-cross} from \Cref{lem:chain-property}, $A$ cross with $X$ if and only if exactly one of $\{C_0, C_\ell\}$ is contained in $X$ but the other one does not. To test so, it suffices to choose an arbitrary vertex from each of $C_0$ and $C_\ell$ and test whether or not it belongs to $X$. This can be done in $O(1)$ time too.

\item[Line~\ref{line:cactus-sub-cases}.]
To test whether $X\setminus A = B$ or not, it suffices to check again if $B\subseteq X$ or $B$ cross with $X$. This can be achieved as described above (implementation of Line~\ref{line:cactus-swap-role-so-A-cross-X}).

\item[\texttt{AddNestingSuperset}.]
It is a bit challenging if we want to search for all mincuts that are currently in $\collection$ that is contained in the given mincut $X$ --- enumerating all vertices in $X$ and then using the disjoint set data structure takes too much time!

To cope with this, 
we handle
minimal mincuts for vertices and edges differently, and describe the implementation details below.

\paragraph{Vertex Case (Line~\ref{line:add-nesting-superset-vertex}).}
Let $X=X_v$ be the minimal mincut of a vertex $v\in V$.
    A cool trick is, we can implement this step utilizing the nesting relation tree $\hat{T}$ defined in \Cref{sec:nesting-relation-tree}, which requires $O(n|\T|)$ preprocessing time by \Cref{lem:cactus-compute-T-hat}.
    
    Let $C_v=\{u_1, u_2, \ldots\}$ be the children of $v$ on $\hat{T}$.
    We notice that upon processing $X_v$, all minimum mincuts of $u_i$ must have been processed already. Hence, by \Cref{inv:cactus-correctness-invariant}, for each $u_i$ there exists some mincut $A_i\in\collection$ that contains $X_{u_i}$.
    Therefore, it suffices to query the disjoint set data structure $|C_v|$ times to identify all mincuts that are covered by $X$.
    
    Since $\hat{T}$ is a tree, there will be exactly $n-1$ $\textsc{Find}$ calls and at most $n-1$ $\textsc{Union}$ calls to the disjoint data structure in total.

\paragraph{Edge Case (Line~\ref{line:add-nesting-superset-edge}).}
Unfortunately all minimal mincuts of edges does not have a hierarchy representation as $\hat{T}$ in the vertex case, so the method we use for the vertex case does not apply to this edge case\footnote{No pun intended.}.
However, this case can be solved easily by making sure we process a bunch of identical minimal mincuts at a time.
Consider the set $F$ of \emph{all} edges $f$ such that $X_f=X_e$.
We claim that the sub-collection of mincuts that contain endpoints to any $f\in F$ is exactly the set of all mincuts to be subsumed. The ``$\Leftarrow$'' direction is trivial, and  the ``$\Rightarrow$'' direction is true because $G[X]$ is connected, and any edge $f\in F$ connecting these mincuts whose minimal mincut has not been processed yet must have $X_f=X_e$.

\end{description}

In conclusion, we successfully proved \Cref{thm:reduction} by providing an algorithm that constructs a cactus in $O(m \alpha(m, n) + n|\T|)$ time. Notice that it is linear time on a not-so-sparse graph whenever $m=\Omega(n|\T|)$ and $|\T|=\Omega(\log n)$.

%% file: B-vertexlabel.tex
\section{Minimal Mincuts of Vertices: Proof of \Cref{lem:labeling-vertices}}
\label{sec:labeling-vertices}

\subsection{On the Missing Case in \cite{karger2009near}}

In \cite{karger2009near}, the key subroutine is computing the minimal mincut for vertices. Under their framework, they compute Type 1, Type 2-Comparable and Type 2-Incomparable minimal 2-respecting mincut for vertices. We believe that their approach is correct, but there seems to be a missing case for computing the minimal incomparable 2-respecting mincut of vertices.

In one of the cases where they compute the minimal incomparable mincut partner for each vertex $v$ (corresponds to \Cref{lem:computing-incomparable-vertex-partner}),
they define the ``outermost'' minimal minprecut parter of $v$. In Section 3.1 of \cite{karger2009near}, two copies of the tree $S$ and $T$ are maintained. One of them is the shrunk tree $S$ where the algorithm contracts processed boughs and produces the bough decomposition.
On the other hand, the algorithm does not shrink $T$ (see \cite[Figure 2]{karger2009near}).
The outermost minprecut is then defined and computed on the uncontracted tree $T$. 
Upon processing the bough $(v_1, v_2, \ldots, v_k)$ where $v_1$ is the lowest vertex, the algorithm examines all edges incident to any vertex in $v_k^\downarrow$ with a postorder traversal, and then dynamically maintains the outermost minimal minprecut partner.
However, in the case where the outermost partner lies in $v_{i}^\downarrow\setminus v_{i-1}^\downarrow$, this partner may become invalid once the algorithm visits $v_i$. Their algorithm did not describe how to update the ``outermost'' partner correctly in this case.

\begin{wrapfigure}{r}{5cm}
\vspace*{-1.5em}
\includegraphics[width=5cm]{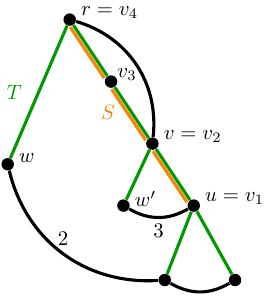}
\caption{A Missing Case.}\label{fig:incomparable minprecut example}
\vspace*{-6em}
\end{wrapfigure}

\paragraph{An Example.} Recall the definition of incomparable precut value $\C^\perp_v(w) = \C(w^\downarrow) - 2\C(v^\downarrow,w^\downarrow)$. Consider the graph in \Cref{fig:incomparable minprecut example} and the green spanning tree $T$, we have
\begin{alignat*}{3}
     &\C_u^\perp(w') &&= \C(w'^\downarrow) - 2\C(u^\downarrow,w'^\downarrow) = 4 - 2\times 3 = -2 ~. \\
     &\C_u^\perp(w) &&= \C(w^\downarrow) - 2\C(u^\downarrow,w^\downarrow) = 3 - 2\times 2 = -1 ~. \\
     &\C_v^\perp(w') &&= \infty~. \qquad (\text{since } w'\parallel v) \\
     &\C_v^\perp(w) &&= \C(w^\downarrow) - 2\C(v^\downarrow,w^\downarrow) = 3 - 2\times 2 = -1 ~.
\end{alignat*}

Therefore, $w$ is the unique outermost minimal minprecut of $v$, and $w'$ is the unique outermost minimal minprecut of $u$. But there doesn't exist any edge between $X = v^\downarrow\setminus u^\downarrow = \{v,w'\}$ and $w^\downarrow$, which turns out to be a missing case for Lemma 3.4 of \cite{karger2009near}. 

Note that maintaining the minprecuts on the uncontracted tree $T$ instead of the shrunk tree $S$ is necessary.
The reason is, it is possible to have a vertex $v$ whose all minprecut partners are processed in boughs of earlier phases. However, if the algorithm shrunk the bough after processing it, 
the minprecut values at vertices on that bough are no longer available.

In the next subsection, we provide a simpler and complete algorithm that uses a completely different approach compared to \cite{karger2009near}.

\subsection{Our Algorithm}
A natural question is this: can we add a self-loop on each vertex and reduce the problem of computing minimal mincut for vertices to computing minimal mincut for edges? The reason we cannot prove this way is that, for the incomparable case, the direction of the reduction is actually opposite. Recall that,  \Cref{lem:labeling-vertex-incomparable} computes the minimal incomparable 2-respecting mincut candidates for vertices, and we use it to prove \Cref{lem:labeling-edges-incomparable} which computes minimal incomparable 2-respecting candidates for edges. Therefore, proving \Cref{lem:labeling-vertex-incomparable} will be one of the main tasks in this section. The key insight is again exploiting the structural property of 2-respecting mincut and using top-tree to find the minimal one. (See the proof of \Cref{lem:computing-incomparable-vertex-partner} at the end of this subsection.)

On the other hand, we can use the self-loop idea to compute the minimal comparable 2-respecting mincut candidates for vertices using \Cref{lem:labeling-edges-comparable}, which does not involve in a circular proof.

\begin{corollary}
\label{cor:labeling-vertex-comparable}
    There is an algorithm that, given a spanning tree $T$ of $G=(V,E)$,
    in total time
    $O(m\log^2 n)$ computes, for every vertex $u\in V$ an comparable 2-respecting minimal mincut candidate $f(u) = (v_u,w_u)$ or $\NULL$ with the following guarantee:
    
    If there exists comparable 2-respecting cut that separating $u$ from root $r$, then $v_u^\downarrow\setminus w_u^\downarrow$ is such a mincut with smallest size.
\end{corollary}
\begin{proof}
    Given a graph $G = (V,E)$, we insert a self-loop edge $e_u$ on each vertex $u\in V$. Denote the new graph as $G' = (V,E')$. By \Cref{lem:labeling-edges-comparable}, given the graph $G'$ and the spanning tree $T$, there is an algorithm computing a comparable 2-respecting minimal mincut candidate $f(e) = (v_e,w_e)$ for each edge $e\in E'$ with the following guarantee: if the minimal mincut of $e$ is a comparable 2-respecting cut of $T$, then $v_e^\downarrow\setminus w_e^\downarrow$ is the minimal mincut of $e$. Therefore, for each vertex $u\in V$, the comparable 2-respecting minimal mincut candidate of $u$ can be set as $f(u) = f(e_u)$, since the minimal mincut of $u$ is the same as the minimal mincut of $e_u$.
\end{proof}

We shall prove \Cref{lem:labeling-vertices} which computes minimal 2-respecting mincut candidates for vertices. %
The proof is essentially the same as \Cref{sec:put-together}, except that we use the corresponding subroutines computing each type of 2-respecting minimal mincut for vertices.

\begin{proof}[Proof of \Cref{lem:labeling-vertices}]
Given a spanning tree $T$, there are three types of 2-respecting cut of $T$. For each type, we compute the minimal mincut candidates for vertices.
\begin{description}[itemsep=0pt]
\item[Type 1.] By \Cref{lem:labeling-vertex-1-respecting}, the algorithm computes the minimal 1-respecting mincut candidate for every vertex $v$.
\item[Type 2-Comparable.] By \Cref{cor:labeling-vertex-comparable}, the algorithm computes the minimal comparable 2-respecting mincut candidates for every vertex $v$.
\item[Type 2-Incomparable.] By \Cref{lem:labeling-vertex-incomparable}, the algorithm computes the minimal incomparable 2-respecting mincut candidates for every vertex $v$.
\end{description}

For each vertex $v$, we get three minimal mincut candidates as above. If all of the three candidates are $\NULL$, then the minimal 2-respecting mincut of $v$ respects to $T$ is $\NULL$. Otherwise, the minimal 2-respecting mincut of $v$ respects to $T$ is the mincut with the smallest size among these minimal mincut candidates.

Since the algorithm for each case runs in $O(m\log^2 n)$ time and the comparing time is constant for each vertex, the whole algorithm runs in $O(m\log^2 n)$ total time.
\end{proof}

In the rest of this section, we focus on proving \Cref{lem:labeling-vertex-incomparable}, which exploits the following lemma as a key subroutine.

\begin{lemma}
\label{lem:computing-incomparable-vertex-partner}
There is an algorithm that, given a graph $G=(V,E)$ and a spanning tree $T\in {\cal T}$, in total time $O(m\log^2 n)$ computes, for every vertex $v\in V$, a vertex called \emph{minimal incomparable mincut partner} $r_v\in V\cup \{\NULL\}$ with the following guarantee: if there exists incomparable mincut partner of $v$, then $r_v$ is the incomparable mincut partner of $v$ with the smallest subtree size $r_v^\downarrow$.
\end{lemma}

With \Cref{lem:computing-incomparable-vertex-partner}, we can efficiently compute minimal incomparable 2-respecting mincut candidates for vertices, since all the minimal incomparable 2-respecting mincut candidates for vertices are of the form $v^\downarrow\cup r_v^\downarrow$ for some $v$. We shall prove \Cref{lem:labeling-vertex-incomparable}, and defer the proof of \Cref{lem:computing-incomparable-vertex-partner} to the end of this subsection. 

\begin{proof}[Proof of \Cref{lem:labeling-vertex-incomparable}]
Observe that if an incomparable 2-respecting mincut contains $u$, it will also contain all the descendants of $u$. Hence, the minimal incomparable 2-respecting mincut of $u$ is either the minimal incomparable 2-respecting mincut of the parent of $u$ or the minimal incomparable 2-respecting mincut $u^\downarrow\cup r_u^\downarrow$. Therefore, we can find the minimal incomparable 2-respecting mincut of all the vertices using a one-time depth-first-search after computing all the $r_v$ in \Cref{lem:computing-incomparable-vertex-partner}.
\end{proof}

The algorithm for \Cref{lem:computing-incomparable-vertex-partner} is the main technical contribution of this subsection. We highlight that top-tree is again the right data structure for exploiting the structural property: using $\MinTreeDown$ we can find the partner with minprecut value and break tie by finding the one with smallest subtree size, which meets exactly the criteria of minimal incomparable mincut partner.

\begin{proof}[Proof of \Cref{lem:computing-incomparable-vertex-partner}]
We will use the reduction to path from \Cref{lem:reduc to path}.
For any $P\in \cP$, define $g(e, P)=r_v$ if $v\in P$ and $e$ is the tree edge with $v$ being the lower vertex, otherwise $g(e, P)=\NULL$.

Given a path $P=(v_1, v_2, \ldots, v_k)$ from the path decomposition with $v_1$ being the deepest vertex, our algorithm will process $v_i$  starting from $i=1, 2, \ldots, k$.
We will maintain the invariant that once we process the vertex $v_i$ the incomparable precut values $\CI_{v_i}(w)$ for all $w \perp v_i$ can be accessed via $val[w]$. %

Next we show how to maintain the invariant. 
In the preprocessing step before the path $P$ was given, we set $val[w] = \C(w^\downarrow)$ for each vertex $w\perp v_1$ and $\infty$ for $w\parallel v_1$, which can be computed in $O(m\log n)$ time.
Now we start from the deepest vertex $v_1$, the algorithm needs to add $-2\C(v_1^\downarrow, w^\downarrow)$ to each $val[w]$ so that $val[w] = \CI_{v_1}(w)$.
To achieve this efficiently, we create a dynamic tree on $T$ (\Cref{lem:data structures}).
For each edge $(u,u')$ where $u\in v_1^\downarrow$,
we invoke $\AddP(u', 2\C(u,u'))$ so that 
two times the weight of the edge $(u,u')$ is added to $val[w]$
for each $w\in u'^\uparrow$.

Then, the algorithm scans through the rest of vertices $v_2,v_3,\cdots v_k$ on the path one by one. Suppose the algorithm reaches $v_i$ now.
With the invariant after processing $v_{i-1}$,
it suffices to add $2\C(v_i^\downarrow\setminus v_{i-1}^\downarrow, w)$ to $val[w]$ for each $w\perp v_i$ by invoking $\AddP(u', 2\C(u, u'))$ for each edge $(u, u')$ where $u\in v_i^\downarrow\setminus v_{i-1}^\downarrow$. (These edges can be found in $O(d(v_i^{\downarrow}\setminus v_{i-1}^\downarrow))$ time using a DFS from $v_i$ without searching the subtree rooted at $v_{i-1}$.) 
Therefore, in $O(d(v_i^\downarrow\setminus v_{i-1}^\downarrow)\log n)$ time, $val[w]$ are updated to $\CI_{v_i}(w)$ for all $w\perp v_i$.

After obtaining $\CI_{v_i}(w)$ values, we compute the minimal incomparable mincut partner $r_{v_i}$ by the following dynamic tree operations. We first invoke $\Cut(v_i,\mathrm{parent}({v_i}))$; then $r_{v_i}$ can be found by $\MinTreeDown(\mathrm{parent}({v_i}))$; finally invoke $\Link(v_i,\mathrm{parent}({v_i}))$ to restore the tree.

From the discussion above, we have an algorithm that, given a path $P$, computes $g(e, P)$ for all $e \in E(P^\downarrow)$ in $O(d(P^\downarrow)\log n)$ time.
By plugging in the path decomposition \Cref{lem:reduc to path}, 
we obtain an algorithm that computes $r_v$ of all vertices $v\in V$ in $O(m\log^2 n)$ total time, because the preprocessing time is $t_p=O(m+n\log n)$ for computing $\C(w^\downarrow)$ and building the dynamic tree.
\end{proof}

%% file: C-omit.tex
\section{The Algorithm from \cite{karger2009near} Runs in $\Omega(m\log^4 n)$ Time}
\label{sec:kp-needs-more-time}

Here, we give an explanation of why the algorithm by \cite{karger2009near} takes $\Omega(m\log^4 n)$ time.  
For the first log factor, their algorithm randomly generates $\Theta(\log n)$ graphs from the original input graph as follows: in each copy, an edge with weight $w$ is contracted with probability $\min\{\frac{w}{2\lambda}, 1\}$. Note that each contracted graph could still contain $\Theta(m)$ edges even if we increase the contraction probability by any constant (e.g.~a complete graph with unit-weight edges). 

Their algorithm then spends at least $\Theta(m \log^3 n)$ time on each contracted graph, which is our time-bound. 
These three log factors come from (1) there are $\Theta(\log n)$ trees in the tree packing, (2) for each tree, there are $\Theta(\log n)$ phases in the bough decomposition, (3) for each phase, we need to process $\Theta(m)$ edges using dynamic tree data structure, each of which takes $\Theta(\log n)$ time. 
Therefore, in total, the algorithm by \cite{karger2009near} takes $\Omega(m\log^4 n)$ time, while ours avoids generating the randomized $\Theta(\log n)$ copies of the graph and takes only $O(m\log^3 n)$ time.

%% file: ms.bbl
\newcommand{\etalchar}[1]{$^{#1}$}
\begin{thebibliography}{HLRW24}

\bibitem[AHLT05]{alstrup2005maintaining}
Stephen Alstrup, Jacob Holm, Kristian~De Lichtenberg, and Mikkel Thorup.
\newblock Maintaining information in fully dynamic trees with top trees.
\newblock {\em Acm Transactions on Algorithms (talg)}, 1(2):243--264, 2005.

\bibitem[BLS20]{BhardwajLovettSandlund2020}
Nalin Bhardwaj, Antonio~Molina Lovett, and Bryce Sandlund.
\newblock A simple algorithm for minimum cuts in near-linear time.
\newblock In Susanne Albers, editor, {\em 17th Scandinavian Symposium and
  Workshops on Algorithm Theory, {SWAT} 2020, June 22-24, 2020, T{\'{o}}rshavn,
  Faroe Islands}, volume 162 of {\em LIPIcs}, pages 12:1--12:18. Schloss
  Dagstuhl - Leibniz-Zentrum f{\"{u}}r Informatik, 2020.

\bibitem[CGL{\etalchar{+}}20]{chuzhoy2020deterministic}
Julia Chuzhoy, Yu~Gao, Jason Li, Danupon Nanongkai, Richard Peng, and
  Thatchaphol Saranurak.
\newblock A deterministic algorithm for balanced cut with applications to
  dynamic connectivity, flows, and beyond.
\newblock In {\em 2020 IEEE 61st Annual Symposium on Foundations of Computer
  Science (FOCS)}, pages 1158--1167. IEEE, 2020.

\bibitem[CLP22a]{cen2022augmenting}
Ruoxu Cen, Jason Li, and Debmalya Panigrahi.
\newblock Augmenting edge connectivity via isolating cuts.
\newblock In {\em Proceedings of the 2022 Annual ACM-SIAM Symposium on Discrete
  Algorithms (SODA)}, pages 3237--3252. SIAM, 2022.

\bibitem[CLP22b]{CLP22b}
Ruoxu Cen, Jason Li, and Debmalya Panigrahi.
\newblock Edge connectivity augmentation in near-linear time.
\newblock In Stefano Leonardi and Anupam Gupta, editors, {\em {STOC} '22: 54th
  Annual {ACM} {SIGACT} Symposium on Theory of Computing, Rome, Italy, June 20
  - 24, 2022}, pages 137--150. {ACM}, 2022.

\bibitem[DEMN21]{DoryEMN21}
Michal Dory, Yuval Efron, Sagnik Mukhopadhyay, and Danupon Nanongkai.
\newblock Distributed weighted min-cut in nearly-optimal time.
\newblock In Samir Khuller and Virginia~Vassilevska Williams, editors, {\em
  53rd Annual {ACM} {SIGACT} Symposium on Theory of Computing (STOC)}, pages
  1144--1153, 2021.

\bibitem[DKL76]{dinits1976structure}
Efim~A. Dinits, Alexander~V. Karzanov, and Micael~V. Lomonosov.
\newblock On the structure of a family of minimum weighted cuts in a graph.
\newblock {\em Studies in Discrete Optimization}, pages 209--306, 1976.

\bibitem[FF09]{Fleiner2009AQP}
Tam{\'a}s Fleiner and Andr{\'a}s Frank.
\newblock A quick proof for the cactus representation of mincuts.
\newblock {\em EGRES Quick-Proofs Series}, 3, 2009.

\bibitem[Fle99]{fleischer1999building}
Lisa Fleischer.
\newblock Building chain and cactus representations of all minimum cuts from
  hao--orlin in the same asymptotic run time.
\newblock {\em Journal of Algorithms}, 33(1):51--72, 1999.

\bibitem[Gab91]{gabow1991applications}
Harold~N Gabow.
\newblock Applications of a poset representation to edge connectivity and graph
  rigidity.
\newblock In {\em [1991] Proceedings 32nd Annual Symposium of Foundations of
  Computer Science}, pages 812--821. IEEE Computer Society, 1991.

\bibitem[Gab06]{gabow2006using}
Harold~N Gabow.
\newblock Using expander graphs to find vertex connectivity.
\newblock {\em Journal of the ACM (JACM)}, 53(5):800--844, 2006.

\bibitem[Gab16]{gabow2016minset}
Harold~N Gabow.
\newblock The minset-poset approach to representations of graph connectivity.
\newblock {\em ACM Transactions on Algorithms (TALG)}, 12(2):1--73, 2016.

\bibitem[GG21]{GeissmannG21}
Barbara Geissmann and Lukas Gianinazzi.
\newblock Parallel minimum cuts in near-linear work and low depth.
\newblock {\em {ACM} Trans. Parallel Comput.}, 8(2):8:1--8:20, 2021.

\bibitem[GH61]{gomory1961multi}
Ralph~E Gomory and Tien~Chung Hu.
\newblock Multi-terminal network flows.
\newblock {\em Journal of the Society for Industrial and Applied Mathematics},
  9(4):551--570, 1961.

\bibitem[GHT18]{goranci2018incremental}
Gramoz Goranci, Monika Henzinger, and Mikkel Thorup.
\newblock Incremental exact min-cut in polylogarithmic amortized update time.
\newblock {\em ACM Transactions on Algorithms (TALG)}, 14(2):1--21, 2018.

\bibitem[GMW20]{gawrychowski2019minimum}
Pawe{\l} Gawrychowski, Shay Mozes, and Oren Weimann.
\newblock Minimum cut in {$ O (m\log^2 n) $} time.
\newblock In {\em 47th International Colloquium on Automata, Languages, and
  Programming, {ICALP}}, 2020.

\bibitem[GT85]{gabow1985linear}
Harold~N Gabow and Robert~Endre Tarjan.
\newblock A linear-time algorithm for a special case of disjoint set union.
\newblock {\em Journal of computer and system sciences}, 30(2):209--221, 1985.

\bibitem[Hen95]{henzinger1995approximating}
Monika~Rauch Henzinger.
\newblock Approximating minimum cuts under insertions.
\newblock In {\em International Colloquium on Automata, Languages, and
  Programming}, pages 280--291. Springer, 1995.

\bibitem[HLRW24]{HLRW2024}
Monika Henzinger, Jason Li, Satish Rao, and Di~Wang.
\newblock Deterministic near-linear time minimum cut in weighted graphs.
\newblock In {\em Proceedings of the 2024 Annual ACM-SIAM Symposium on Discrete
  Algorithms (SODA)}, 2024.

\bibitem[HO92]{hao1992faster}
Jianxiu Hao and James~B. Orlin.
\newblock A faster algorithm for finding the minimum cut in a graph.
\newblock In {\em Proceedings of the Third Annual {ACM/SIGACT-SIAM} Symposium
  on Discrete Algorithms (SODA)}, 1992.

\bibitem[HRW20]{henzinger2020local}
Monika Henzinger, Satish Rao, and Di~Wang.
\newblock Local flow partitioning for faster edge connectivity.
\newblock {\em SIAM Journal on Computing}, 49(1):1--36, 2020.

\bibitem[Kar93]{karger1993global}
David~R Karger.
\newblock Global min-cuts in rnc, and other ramifications of a simple min-cut
  algorithm.
\newblock In {\em SODA}, volume~93, pages 21--30. Citeseer, 1993.

\bibitem[Kar00]{karger2000minimum}
David~R Karger.
\newblock Minimum cuts in near-linear time.
\newblock {\em Journal of the ACM (JACM)}, 47(1):46--76, 2000.

\bibitem[KP09]{karger2009near}
David~R Karger and Debmalya Panigrahi.
\newblock A near-linear time algorithm for constructing a cactus representation
  of minimum cuts.
\newblock In {\em Proceedings of the Twentieth Annual ACM-SIAM Symposium on
  Discrete Algorithms}, pages 246--255. SIAM, 2009.

\bibitem[KS96]{karger1996new}
David~R Karger and Clifford Stein.
\newblock A new approach to the minimum cut problem.
\newblock {\em Journal of the ACM (JACM)}, 43(4):601--640, 1996.

\bibitem[KT86]{karzanov1986efficient}
Alexander~V Karzanov and Eugeniy~A Timofeev.
\newblock Efficient algorithm for finding all minimal edge cuts of a
  nonoriented graph.
\newblock {\em Cybernetics}, 22(2):156--162, 1986.

\bibitem[KT18]{kawarabayashi2018deterministic}
Ken{-}ichi Kawarabayashi and Mikkel Thorup.
\newblock Deterministic edge connectivity in near-linear time.
\newblock {\em J. {ACM}}, 66(1):4:1--4:50, 2018.

\bibitem[Li21]{li2021deterministic}
Jason Li.
\newblock Deterministic mincut in almost-linear time.
\newblock In {\em Proceedings of the 53rd Annual ACM SIGACT Symposium on Theory
  of Computing}, pages 384--395, 2021.

\bibitem[LP20]{li2020deterministic}
Jason Li and Debmalya Panigrahi.
\newblock Deterministic min-cut in poly-logarithmic max-flows.
\newblock In {\em 2020 IEEE 61st Annual Symposium on Foundations of Computer
  Science (FOCS)}, pages 85--92. IEEE, 2020.

\bibitem[LST20]{lo2020compact}
On-Hei Lo, Jens Schmidt, and Mikkel Thorup.
\newblock Compact cactus representations of all non-trivial min-cuts.
\newblock {\em Discrete Applied Mathematics}, 303, 04 2020.

\bibitem[MN20]{Mukhopadhyay2020WeightedMS}
Sagnik Mukhopadhyay and Danupon Nanongkai.
\newblock Weighted min-cut: sequential, cut-query, and streaming algorithms.
\newblock {\em Proceedings of the 52nd Annual ACM SIGACT Symposium on Theory of
  Computing}, 2020.

\bibitem[NGM97]{naor1997fast}
Dalit Naor, Dan Gusfield, and Charles Martel.
\newblock A fast algorithm for optimally increasing the edge connectivity.
\newblock {\em SIAM Journal on Computing}, 26(4):1139--1165, 1997.

\bibitem[NI92a]{nagamochi1992computing}
Hiroshi Nagamochi and Toshihide Ibaraki.
\newblock Computing edge-connectivity in multigraphs and capacitated graphs.
\newblock {\em SIAM Journal on Discrete Mathematics}, 5(1):54--66, 1992.

\bibitem[NI92b]{nagamochi1992linear}
Hiroshi Nagamochi and Toshihide Ibaraki.
\newblock A linear-time algorithm for finding a sparse k-connected spanning
  subgraph of ak-connected graph.
\newblock {\em Algorithmica}, 7(1):583--596, 1992.

\bibitem[NK94]{Nagamochi1994CanonicalCR}
Hiroshi Nagamochi and Tiko Kameda.
\newblock Canonical cactus representation for minimum cuts.
\newblock {\em Japan Journal of Industrial and Applied Mathematics},
  11:343--361, 1994.

\bibitem[NNI00]{nagamochi2000fast}
Hiroshi Nagamochi, Yoshitaka Nakao, and Toshihide Ibaraki.
\newblock A fast algorithm for cactus representations of minimum cuts.
\newblock {\em Japan journal of industrial and applied mathematics},
  17(2):245--264, 2000.

\bibitem[NV91]{naor1991representing}
Dalit Naor and Vijay~V Vazirani.
\newblock Representing and enumerating edge connectivity cuts in rnc.
\newblock In {\em Workshop on Algorithms and Data Structures}, pages 273--285.
  Springer, 1991.

\bibitem[RZZ23]{ravi2023approximation}
R~Ravi, Weizhong Zhang, and Michael Zlatin.
\newblock Approximation algorithms for steiner tree augmentation problems.
\newblock In {\em Proceedings of the 2023 Annual ACM-SIAM Symposium on Discrete
  Algorithms (SODA)}, pages 2429--2448. SIAM, 2023.

\bibitem[Sar21]{saranurak2021simple}
Thatchaphol Saranurak.
\newblock A simple deterministic algorithm for edge connectivity.
\newblock In {\em Symposium on Simplicity in Algorithms (SOSA)}, pages 80--85.
  SIAM, 2021.

\bibitem[ST83]{sleator1983data}
Daniel~D Sleator and Robert~Endre Tarjan.
\newblock A data structure for dynamic trees.
\newblock {\em Journal of computer and system sciences}, 26(3):362--391, 1983.

\bibitem[SW97]{stoer1997simple}
Mechthild Stoer and Frank Wagner.
\newblock A simple min-cut algorithm.
\newblock {\em Journal of the ACM (JACM)}, 44(4):585--591, 1997.

\bibitem[SY23]{SaranurakY23}
Thatchaphol Saranurak and Sorrachai Yingchareonthawornchai.
\newblock Deterministic small vertex connectivity in almost linear time, 2023.
\newblock To appear at SODA'23.

\end{thebibliography}
